\documentclass{article}

\usepackage[utf8]{inputenc}
\usepackage{setspace}
\usepackage[rgb]{xcolor}
\usepackage{verbatim}
\usepackage{amsgen,amsmath,amstext,amsbsy,amsopn,tikz,amssymb,amsthm}
\usepackage{fancyhdr}
\usepackage{rotating}
\usepackage{enumitem}
\usepackage{graphicx}
\usepackage{mathtools}
\usepackage{commath}
\usepackage{pdfpages}
\usepackage{breqn}
\usepackage{multicol}
\usepackage[title]{appendix}

\newtheorem{theorem}{Theorem}[section]
      
      \newtheorem{lemma}[theorem]{Lemma}
      \newtheorem{corollary}[theorem]{Corollary}
      \newtheorem{definition}[theorem]{Definition}
      
      \newtheorem{remark}[theorem]{Remark}

      \numberwithin{equation}{section}

\usepackage{booktabs}
\usepackage[backend=biber, style=numeric , maxnames=10, sorting=nty]{biblatex}
\newcommand{\al}{\alpha}

\newcommand{\de}{\delta}
\newcommand{\sq}{\sqrt}

\newcommand{\kct}{{k\choose{2}}}
\newcommand{\st}{\Tilde{s_2}}
\newcommand{\dt}{\Tilde{d_2}}
\newcommand{\stg}{\frac{2-\de}{2+\de}}

\DeclareMathOperator*{\E}{\mathbb{E}}

\begin{document}


\title{A tight bound for the clique query problem in two rounds}

\author{Uriel Feige~\thanks{Weizmann Institute, uriel.feige@weizmann.ac.il}
\and Tom Ferster~\thanks{This work was done while the author was a student at the Weizmann Institute, \mbox{tomferster@gmail.com}}}
\maketitle{}

\begin{abstract}

We consider a problem introduced by Feige, Gamarnik, Neeman, Rácz and Tetali [2020], that of finding a large clique in a random graph $G\sim G(n,\frac{1}{2})$, where the graph $G$ is accessible by queries to entries of its adjacency matrix.  The query model allows some limited adaptivity, with a constant number of rounds of queries, and $n^\de$ queries in each round. With high probability, the maximum clique in $G$ is of size roughly $2\log n$, 
and the goal is to find cliques of size $\al\log n$, for  {$\al$} as large as possible. We prove that no two-rounds algorithm is likely to find a clique larger than $\frac{4}{3}\de\log n$, which is a tight upper bound when $1\leq\de\leq \frac{6}{5}$. For other ranges of parameters, namely, two-rounds with $\frac{6}{5}<\de<2$, and three-rounds with $1\leq\de<2$, we improve over the previously known upper bounds on $\al$, but our upper bounds are not tight.  
If early rounds are restricted to have fewer queries than the last round, then for some such restrictions we do prove tight upper bounds. 
\end{abstract}


\section{Introduction}

We consider a problem introduced by Feige, Gamarnik, Neeman, Rácz and Tetali \cite{feige2020finding}, that of finding a large clique in a random graph $G\sim G(n,\frac{1}{2})$, where the graph $G$ is accessible by queries to entries of its adjacency matrix. The input to the problem is a random graph $G(V,E)$ on $n$ vertices, where for every pair $\{u,v\}$ of distinct vertices,  $(u,v)\in E$ independently with probability $\frac{1}{2}$. To simplify some of the terminology, we sometimes refer to a pair $\{u,v\}$ of distinct vertices as an edge, regardless of whether $(u,v)\in E$, but if $(u,v)\in E$ we refer to it as a positive edge, and if $(u,v)\not\in E$ we refer to it as a negative edge. The goal of the algorithm is to find in $G$ a clique (a complete subgraph) that is as large as possible. The algorithm can access $G$ via queries to the adjacency matrix of $G$, querying an edge $\{u,v\}$ and getting in reply its status, whether it is positive or negative. We may think of these queries as queries to a random oracle, where for each new query (not previously asked), the oracle answers positively with probability $\frac{1}{2}$ and negatively with probability $\frac{1}{2}$, independently of all previous answers. The limiting resource for the algorithm is the total number of queries $q$, whereas the running time of the algorithm may be unbounded (e.g., exponential in $n$). The sequence of queries asked by the algorithm may be adapative, in the sense that the choice of the next query depends on answers to previous queries. 

The results of the $q$ queries induce a graph $G'(V,E')$, where the edges of $G'$ are those edges that were queried and received a positive answer. As the algorithm is not bounded in its computation time, it can find the maximum clique in $G'$. Hence the goal of the querying process is to produce a graph $G'$ in which the maximum clique is as large as possible. We denote this maximum clique by $K$ and its size by $k$. Note that for any nontrivial algorithm (for which $k > 1$), $K$ (and likewise $k$) is a random variable that will depend on the input graph $G$. Hence for a given algorithm, we shall think of $k$ as the median value of the size of the clique that the algorithm finds. For reasons that will become clear later, if $k$ is even our results are somewhat simpler to express. Hence, we shall assume that $k$ is even (this has negligible effect on the results, a difference of at most~1 in the size of the clique found).

We shall use $\delta$ as a parameter indicating the allowed number of queries, where $q = n^{\delta}$. (Sometimes the interpretation of $\delta$ will be relaxed to requiring $q = O(n^{\delta})$. This has negligible effect on the results.) Hence our goal is to maximize $k$ (as a function of $n$ and $\delta$), such that there is an algorithm making at most $n^\de$ queries, for which with probability at least $\frac{1}{2}$ over the choice of $G\sim G(n,\frac{1}{2})$, the graph $G'$ contains a clique $K$ of size $k$.

We consider only $\de\geq 1$, as for smaller $\de$ the number of non-isolated vertices in $G'$ will be $n'<n$, and we may solve a similar problem with $n'$ instead of $n$. For $\de=2$, the whole graph may be queried. The clique of maximum size in a graph $G\sim G(n,\frac{1}{2})$ is with high probability (w.h.p.) of size roughly (neglecting low order terms) $2\log n$ (see  \cite{matula1972employee}). {Therefore, the size of the maximum clique in $G'$ satisfies w.h.p. $k\leq 2\log n + \Theta(1)$ for every $1\leq\de<2$. }On the other hand, a clique of (expected) size $\log n$ may be found using roughly {$2n$} queries in the following way. Initialize the sets $K=\emptyset$, $T=V(G)$. While $T$ is not empty repeat the following:
\begin{enumerate}
    \item Choose arbitrarily $t\in T$, and update $K\leftarrow K\cup \{t\}$.
    \item Query all pairs $\{\{t,t'\}\mid t'\in T\setminus\{t\}\}$.
    \item Remove $t$ and all its non-neighbors from $T$.
\end{enumerate}
The size of $T$ decreases by a factor of $2$ in each iteration (in expectation), hence there are roughly $\log n$ iterations, and the total expected number of queries is {$2n$}. After the last iteration, the set $K$ is a clique of expected size $\log n$. {Therefore, the values of $k$ of interest are between $\log n$ and $2\log n$.} In all computations we will neglect terms of order $o(\log n)$ in the size of $k$. Hence, we can state the problem we wish to solve in the following way: Maximize $\al$ for which a clique of size $k\ge \al\log n - o(\log n)$ may be found using a sequence of at most $n^\de$ queries. Note that $\al$ is a function of $\de$ which satisfies $1\leq\al\leq 2$. We give a more formal definition of the value $\al$ we wish to find.
\begin{definition}
\label{def_org_1}
     Let $1\leq\de<2$, $\al\leq 2$. A deterministic algorithm $\mathcal{A}$, which makes a sequence of at most $n^\de$ adaptive queries (based on oracle answers), will be called a $(\de,\al)$-algorithm {if for every sufficiently large $n$, with probability at least $\frac{1}{2}$, the graph $G'$ (induced on edges discovered by the algorithm) contains a clique of size $k\ge \al\log n - o(\log n)$.} 
\end{definition}
\begin{definition}
\label{def_org_2}
    Let $1\leq\de<2$. $\al_{*}(\de)$ is defined as the supremum over $\al$ such that there exists a $(\de,\al)$-algorithm. Note that $1\leq \al_{*}(\de)\leq 2$.
\end{definition}
\begin{remark}
{Since we do not limit the time complexity of algorithms, if $G'$ contains at least one clique of size $k$, $\mathcal{A}$ can be assumed to find such a clique $K$ (e.g., by using exhaustive search), and to return it.}
\end{remark}

A full solution to the stated problem is to find the exact value of $\al_{*}(\de)$ for every $1\leq \de<2$. Lower bounds on $\al_{*}(\de)$ are proved by algorithms which find cliques of size $\al\log n$ (for $\al>1$). Such algorithms were presented in \cite{feige2020finding}, and can be found in Section \ref{algos_sec}. The presented algorithms are suitable also for a slightly different problem setting, in which the adaptiveness of the algorithm is restricted to a constant number of rounds. The upper bounds we will present in this paper will be for this restricted setting. In this setting, the algorithm may use a constant number $l$ of rounds, where in each round it may query at most $n^\de$ edge queries simultaneously, and the oracle answers all of them. This means that the algorithm may choose queries on round $i$ based on answers given in previous rounds, but not on those given for queries in round $i$. Denote by $q_i$ the number of queries in round $i$. The pattern of queries in the first round gives us the queries graph $Q_1$ with $n$ vertices and $q_1$ edges. Likewise, the pattern of queries in any other round $i$ (that may depend on the answers to queries of previous rounds) gives us the queries graph $Q_i$ with $n$ vertices and $q_i$ edges. {Note that the edges of $Q_i$ consist of all pairs queried in round $i$, and not only positive ones.} {We denote by $Q_{\leq i}$ the queries graph which includes all queries up to round $i$, i.e., $E(Q_{\leq i})=\bigcup_{j\leq i} E(Q_j)$. }(Note that we use the term 'up to' in the meaning of 'up to and including'). We modify definitions \ref{def_org_1} and \ref{def_org_2} to comply with the restricted setting.

\begin{definition}
\label{def_rest_1}
     Let $l\in\mathbb{N}$, $1\leq\de<2$, $\al\leq 2$. A deterministic algorithm $\mathcal{A}$, which makes $l$ rounds of queries with at most $n^\de$ queries in each round, will be called a $(\de,l,\al)$-algorithm {if for every sufficiently large $n$, with probability at least $\frac{1}{2}$, $G'$ contains a clique of size $k\ge \al\log n - o(\log n)$.} 
\end{definition}
\begin{definition}
\label{def_rest_2}
    Let $l\in\mathbb{N}$, $1\leq\de<2$. $\al_{*}(\de,l)$ is defined as the supremum over $\al$ such that there exists a $(\de,l,\al)$-algorithm. 
\end{definition}
These definitions can be generalized to the case where the number of queries may not be the same in all rounds. In general, we may assume that in round $i$ the algorithm is restricted to $q_i=n^{\de_i}$ queries for some $0<\de_i<2$. This leads to similar definitions of a $((\de_1,\dots,\de_l),l,\al)$-algorithm and $\al_{*}((\de_1,\dots,\de_l),l)$ (see definitions in Section \ref{sec_two_rounds_large_delta}). Note that for every $l\in \mathbb{N}$, denoting $\de\coloneqq \max_i \de_i$ we have:
\[\al_{*}((\de_1,\dots,\de_l),l)\leq \al_{*}(\de,l)\leq \al_{*}(\de)\]

\subsection{Previous results}
\label{prev_res}

The problem of bounding $\al_{*}(\de)$ and $\al_{*}(\de,l)$ was studied by Feige, Gamarnik, Neeman, Rácz and Tetali \cite{feige2020finding}. They proved the following results regarding $\al_{*}(\de,l)$:

\begin{enumerate}
    \item For every $1\leq \de<2$,  $\al_{*}(\de,1)=\de$. 
    \item For every $1\leq\de\leq\frac{6}{5}$, $\al_{*}(\de,2)\geq\frac{4\de}{3}$, and for every $\frac{6}{5}\leq\de<2$, $\al_{*}(\de,2)\geq 1+\frac{\de}{2}$.
    \item For every $1\leq \de<2$, $\al_{*}(\de,3)\geq1+\frac{\de}{2}$.
    \item {For $\de=1$ and every $l\in\mathbb{N}$, $\al_{*}(1,l)\leq 2^{1-\frac{1}{2^l-1}}$. This gives: $\al_{*}(1,2)\leq 1.588$, $\al_{*}(1,3)\leq 1.812$ and so on.}
    \item For every $1\leq\de<2$ and $l\in\mathbb{N}$,  $\al_{*}(\de,l)\leq 2-\de\left(\frac{2-\de}{2}\right)^l$.
    
\end{enumerate}

An upper bound on $\al_{*}(\de)$ was proved by Alweiss, Ben Hamida, He, and Moreira \cite{alweiss2021subgraph}. They showed that for every $1\leq\de<2$:
\[\al_{*}(\de)\leq 1+\sqrt{1-\frac{(2-\de)^2}{2}}\]
For example, for $\de=1$ this gives $\al_{*}(1)\leq 1+\sqrt{\frac{1}{2}}\sim 1.707$. Since $\al_{*}(\de,l)\leq \al_{*}(\de)$, {this result also improves the upper bound on $\al_{*}(\de,l)$ for every $l\geq 3$.}

\section{Our Results}

We prove upper bounds on $\al_{*}(\de,l)$ and $\al_{*}((\de_1,\dots,\de_l),l)$ for $l\in\{2,3\}$ (see Figure \ref{bounds_fig}). Namely, we upper bound the size of clique that is likely to be found using two or three rounds algorithms. First, we prove the following upper bound for two-rounds algorithms: 
\begin{theorem}
\label{thm_results_1}
For every $1\leq\de<2$:
\[\al_{*}(\de,2)\leq \frac{4}{3}\de\]
\end{theorem}
The upper bound in Theorem \ref{thm_results_1} is tight when $1\leq \de \leq \frac{6}{5}$. When $\frac{6}{5}<\de<2$ we have the lower bound $\al_{*}(\de,2)\geq 1+\frac{\de}{2}$ (see \cite{feige2020finding} or Section \ref{algos_sec}), so there is a gap between the bounds. We improve the upper bound for this case by proving the following theorem:

\begin{theorem}
\label{thm_results_2}
For every $\frac{6}{5}<\de<2$:
\[\al_{*}(\de,2)\leq 1+\sqrt{(\de-1)(3-\de)}\]
\end{theorem}
The lower bound $\al_{*}(\de,2)\geq 1+\frac{\de}{2}$, when $\frac{6}{5}<\de<2$, is achieved by an algorithm which queries less than $n^{\de}$ queries in round 1. Therefore, we consider algorithms in which round 1 is restricted to query at most $n^{\de_1}$ pairs for some $\de_1<\de$. We prove the following theorem:

\begin{theorem}
\label{thm_results_3}
For every $\frac{6}{5}<\de<2$:
\[\al_*\Big(\Big(\frac{3}{2}-\frac{\de}{4},\de\Big),2\Big)\leq 1+\frac{\de}{2}\]
\end{theorem}
The upper bound in Theorem \ref{thm_results_3} is tight for every $\frac{6}{5}<\de<2$.

For three-rounds algorithms, we have the the lower bound $\al_{*}(\de,3)\geq 1+\frac{\de}{2}$ (see \cite{feige2020finding} or Section \ref{algos_sec}). This lower bound is achieved by an algorithm which queries less than $n^{\de}$ queries in rounds 1 and 2. We consider algorithms in which the first two rounds are restricted, and prove the following theorem:

\begin{theorem}
\label{thm_results_4}
For every $1\leq\de<2$:
\[\al_*\Big(\Big(1-\frac{\de}{2},1,\de\Big),3\Big)\leq 1+\frac{\de}{2}\]
\end{theorem}
The upper bound in Theorem \ref{thm_results_3} is tight for every $1\leq \de<2$. We also prove upper bounds on $\al_{*}(\de,3)$ (i.e., when no round is restricted), for several values of $\de$. The bounds we prove are summarized in the following table:

\begin{table}[ht]
\begin{center}
\begin{tabular}{ |c||c c c c c c c c c c| } 
 \hline
 $\de$ & 1 & 1.1 & 1.2 & 1.3 & 1.4 & 1.5 & 1.6 & 1.7 & 1.8 & 1.9  \\ 
 $\al_*(\de,3)\leq$ & 1.62 & 1.69 & 1.77 & 1.83 & 1.88 & 1.92 & 1.95 & 1.97 & 1.99 & 1.997\\ 

 \hline
\end{tabular}
\end{center}
\caption{Upper bounds on $\al_*(\de,3)$}
\end{table}

\begin{figure}[h!]
  \centering
  \includegraphics[width=70mm, scale=0.5]{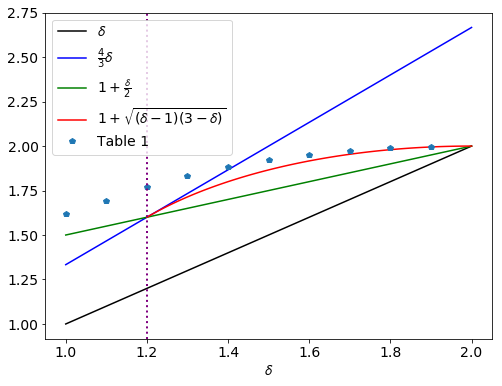}
  \caption{An illustration of upper and lower bounds on $\al_{*}(\de,l)$ and $\al_{*}((\de_1,\dots,\de_l),l)$ for $l\in\{1,2,3\}$}
  \label{bounds_fig}
\end{figure}

For every $1\leq\de<2$, the upper bounds we prove on $\al_{*}(\de,2)$ and $\al_{*}(\de,3)$ are stronger than the upper bounds on $\al_{*}(\de,l)$ proved in \cite{feige2020finding}, and than the upper bound $\al_{*}(\de)\leq 1+\sqrt{1-\frac{(2-\de)^2}{2}}$ proved in \cite{alweiss2021subgraph}.

\section{Related Work}

There is extensive research on the query complexity of deterministic and randomized algorithms. For a given graph property, the query complexity is the smallest number of edge queries needed to determine whether an input graph satisfies the property. For some properties, such as determining whether a graph contains an edge, or whether a graph is connected, 
all edges must be queried. These are called evasive properties. The Aanderaa–Karp–Rosenberg conjecture claims that every nontrivial monotone graph property (i.e remains true when edges are added) for graphs on $n$ vertices is evasive. While the AKR conjecture is still open, there are many works proving variations of it. Rivest and Vuillemin \cite{rivest1975generalization} showed that at least $\frac{n^{2}}{16}$ queries are needed to test for any nontrivial monotone graph property. The constant was later improved by Scheidweiler and Triesch\cite{scheidweiler2013lower} to the current best result of $\frac{n^{2}}{3}-o(n^{2})$ queries. Kahn, Saks and Sturtevant used a topological approach to show that the AKR conjecture is true when the number of vertices is a prime power \cite{kahn1984topological}, and Yao showed it in the case of bipartite graphs \cite{yao1988monotone}.

There has also been research on the query complexity of randomized algorithms. No nontrivial monotone property is known with randomized query complexity less than $\frac{n^{2}}{4}$, and a similar conjecture by Karp says that $\Omega(n^{2})$ queries are required in this case. The first non-linear lower bound, of $\Omega(n\log^{1/12}n)$, was proved by Yao \cite{yao1991lower}, and later improved to the current best lower bound of $\Omega(n^{4/3}\log^{1/3}n)$ by Chakrabarti and Khot \cite{chakrabarti2001improved}. Other lower bounds were given by Friedgut, Kahn and Wigderson \cite{friedgut2002computing} and O'Donnell, Saks, Schramm and Servedio \cite{o2005every} where the bound depends on a critical probability $p$ such that a graph drawn from $G(n,p)$ satisfies the property with probability $\frac{1}{2}$. The specific property of having a subgraph isomorphic to a given graph (which is similar to our case) was proved by Gröger to have a lower bound of $\Omega(n^{3/2})$ \cite{dietmar1992randomized}. A summary on Evasiveness of Graph Properties can be found in lecture notes by Neal Young from a course taught by László Lovász \cite{lovasz2002lecture}.

A setting more similar to ours, where the input graph is drawn randomly from $G(n,p)$, was introduced by Ferber, Krivelevich, Sudakov and Viera. They proved that when $p>\frac{\ln n+\ln\ln n+\omega(1)}{n}$ a Hamiltonian cycle can be found using $(1+o(1))n$ positive edge queries (which is optimal) \cite{ferber2016finding}. However, when $p=\frac{1+\epsilon}{n}$, in order to find a path of length $l=\Omega(\frac{\log(\frac{1}{\epsilon})}{\epsilon})$ at least $\Omega(\frac{l}{p\epsilon\log(\frac{1}{\epsilon})})$ queries are required (and not only the trivial bound of $\Omega(\frac{l}{p})$) \cite{ferber2017finding}. Colon, Fox, Grinshpun and He \cite{conlon2019online}, as part of  work concerning the online Ramsey number, studied the problem of determining the number of queries needed to find in $G\sim G(n,p)$ a copy of a target graph, in particular, a constant size clique.

The problem we study, of finding a large clique in $G\sim G(n,\frac{1}{2})$, 
was introduced by Feige, Gamarnik, Neeman, Rácz and Tetali \cite{feige2020finding}, 
and further studied by Alweiss, Ben Hamida, He, and Moreira \cite{alweiss2021subgraph}. 
See Section \ref{prev_res} for results from these papers. 
The latter paper also proves bounds on the number of queries to a random graph $G\sim G(\infty,p)$ needed in order to find a copy of a graph $H$ with degeneracy $d$.

A related problem, that of finding a planted clique of size $k\geq (2+\epsilon)\log n$ in a $G(n,\frac{1}{2})$ graph by adaptive probing, was studied by Ràcz and Schiffer \cite{racz2019finding}.  They proved that $o(\frac{n^{2}}{k^{2}}+n)$ queries do not suffice in order to find the planted clique, and that $O(\frac{n^{2}}{k^{2}}\log^{2}n+n\log n)$ queries do suffice. Mardia, Asi and Chandrasekher \cite{mardia2020finding} designed sub-linear time algorithms {($o(n^2)$ time complexity, and consequently also $o(n^2)$ queries)} that find planted cliques of size {$k=\omega(\sq{n\log\log n})$}.

\section{Preliminaries}
\label{prem}

\subsection{Algorithms - lower bounds}
\label{algos_sec}
The following algorithms, which lower bound $\al_{*}(\de,l)$, were presented in \cite{feige2020finding}.

\subsubsection*{Algorithm 1 - One round} 

The following one round algorithm finds a clique of size $\de\log n$:

\textbf{Round 1:} Choose arbitrarily a set $S$ of $n^{\de/2}$ vertices, and query all pairs within $S$ (roughly $n^\de$ pairs). Find a clique $K$ of size $\de\log n$ in the graph induced on $S$, and return $K$. Such a clique exists w.h.p. as this is the expected size of the maximum clique in $S$.


\subsubsection*{Algorithm 2 - Two rounds when $\de\leq \frac{6}{5}$} 

The following two rounds algorithm finds a clique of size $\frac{4}{3}\de\log n$ when $1\leq \de\leq\frac{6}{5}$:

\textbf{Round 1:} Choose arbitrarily a set $S$ of $n^{\de/6}$ vertices and a set $T$ of $n^{{5\de}/6}$ vertices not in $S$ ($|T|\leq n$ since $\de\leq\frac{6}{5}$). Query all pairs within $S$ and all pairs in $S\times T$ (roughly $n^\de$ pairs). Find a clique $S'$ of size $\frac{\de}{3}\log n$ in the graph induced on $S$. Such a clique exists w.h.p. as this is the expected size of the maximum clique in $S$. Let $T'\subset T$ be the set of mutual neighbors of $S'$. It is of expected size $|T|\cdot 2^{-\frac{\de}{3}\log n}=n^{{\de}/2}$.

\textbf{Round 2:} Query all pairs within $T'$ (roughly $n^\de$ pairs). Find a clique $T''$ of size $\de\log n$ in the graph induced on $T'$ (neglecting low order terms, such a clique exists w.h.p.). Return $K\coloneqq S'\cup T''$, which is of size $\frac{4\de}{3}\log n$.

We get that for every $1\leq \de\leq \frac{6}{5}$: $\al_{*}(\de,2)\geq \frac{4}{3}\de$.

\subsubsection*{Algorithm 3 - Two rounds when $\de\geq \frac{6}{5}$} 

The following two rounds algorithm finds a clique of size $(1+\frac{\de}{2})\log n$ when $\frac{6}{5}\leq\de<2$:

\textbf{Round 1:} Choose arbitrarily a set $S$ of $n^{(1-\de/2)/2}$ vertices, and let $T$ be the set of all vertices not in $S$. Query all pairs within $S$ (roughly $n^{1-\de/2}<n^\de$ pairs), and all pairs in $S\times T$ (roughly $n^{3/2-\de/4}$ pairs, which is less than $n^\de$ when $\de\geq \frac{6}{5}$). Find a clique $S'$ of size $(1-\frac{\de}{2})\log n$ in the graph induced on $S$ (exists w.h.p.). Let $T'\subset T$ be the set of mutual neighbors of $S'$. It is of expected size $|T|\cdot 2^{-(1-\frac{\de}{2})\log n}\simeq n^{{\de}/2}$. 

\textbf{Round 2:} Query all pairs within $T'$ (roughly $n^\de$ pairs). Find a clique $T''$ of size $\de\log n$ (neglecting low order terms) in the graph induced on $T'$ (exists w.h.p.). Return $K\coloneqq S'\cup T''$, which is of size $(1+\frac{\de}{2})\log n$.

We get that for every $ \frac{6}{5}\leq \de <2$: $\al_{*}(\de,2)\geq 1+\frac{\de}{2}$.

\subsubsection*{Algorithm 4 - Three rounds} 

The following three rounds algorithm finds a clique of size $(1+\frac{\de}{2})\log n$.

\textbf{Round 1:} Choose arbitrarily a set $S$ of $n^{(1-\de/2)/2}$ vertices. Query all pairs within $S$ (roughly $n^{1-\de/2}<n^\de$ pairs). Find a clique $S'$ of size $(1-\frac{\de}{2})\log n$ in the graph induced on $S$ (exists w.h.p.).

\textbf{Round 2:} Let $T$ be a set of $\frac{n}{|S'|}$ vertices not in $S$. Query all pairs in $S'\times T$ (roughly $n$ pairs) and let $T'\subset T$ be the set of mutual neighbors of $S'$. It is of expected size $|T|\cdot 2^{-(1-\frac{\de}{2})\log n}\simeq \frac{n^{{\de}/2}}{|S'|}$.

\textbf{Round 3:} Query all pairs within $T'$ (roughly $n^\de$ pairs). Find a clique $T''$ of size $\de\log n$ (neglecting low order terms) in the graph induced on $T'$ (exists w.h.p.). Return $ K\coloneqq S'\cup T''$, which is of size $(1+\frac{\de}{2})\log n$.

We get that for every $1\leq \de<2$: $\al_{*}(\de,3)\geq 1+\frac{\de}{2}$.

\bigskip

Algorithms 2 and 3 have the same queries structure in general. Both query on round 1 all pairs within some set $S$, and between $S$ and some set $T$, and on round 2 query all pairs within some $T'\subset T$. The sizes of the sets are chosen such that the constraint on the number of queries is satisfied, and the returned clique size is maximized (in relation to similar algorithms). The major difference between algorithms 2 and 3 is that in Algorithm 2 the set of vertices participating in any query is of size smaller than $n$ (when $\de<\frac{6}{5}$). So the effective number of vertices in $G$ in that case is smaller than $n$ (as we may discard vertices that are never queried). When $\de=\frac{6}{5}$ both algorithms are identical and this is the smallest $\de$ for which the size of $T$ (and therefore the effective size of the graph $G$) is of size roughly $n$. For larger $\de$, the size of $T$ cannot increase further, and the growth rate of $\al$ as a function of $\de$ decreases (from $\frac{4}{3}\de$ to $1+\frac{\de}{2}$).

In fact, Algorithm 3 resembles Algorithm 4 more than it resembles Algorithm 2. Both algorithms 3 and 4 find the same clique size, and the extra round in Algorithm 4 has no benefit to the size of the clique. The only difference is that round 1 of Algorithm 3 is split to rounds 1,2 in Algorithm 4, and this allows to leave out some unnecessary queries. This reduction in the number of queries allows using Algorithm 4 also when $\de<\frac{6}{5}$, in which case the clique returned is larger than the one returned by Algorithm 2 ($\frac{4}{3}\de<1+\frac{\de}{2}$ for $\de<\frac{6}{5}$).

\subsection{Constrained algorithms}
\label{constrained algs}
Let $\mathcal{A}$ be a $(\de,l,\al)$-algorithm. Assuming the algorithm succeeds, we denote the clique it returns by $K$, where $|V(K)|=k=\al \log n$. After round $i$, each pair of vertices $u,v\in V$ is of one of three types:
\begin{enumerate}
    \item The pair $u,v$ was not queried up to round $i$.
    \item The pair $u,v$ was queried up to round $i$ and answered positively.
    \item The pair $u,v$ was queried up to round $i$ and answered negatively.
\end{enumerate}
The state $K_i$, of the returned clique after round $i$, is the state of all pairs of vertices of $K$ after round $i$. Since we assume the algorithm succeeds and $K$ is indeed a clique, the last state $K_l$ consists only of pairs of type $2$. Therefore, there are no pairs of type $3$ in any of the states $K_i$, and we can look at $K_i$ as a subgraph of $K_l$ which contains only edges queried up to round $i$. Note that $K_0$ is an independent set.

The only constraints we had on the algorithm so far were the number of query rounds it was allowed to make and the number of edge queries in each of them. Following the proof method introduced in \cite{feige2020finding}, the algorithms we will analyze will have some more constraints upon them. Those constraints will refer to properties of the states $K_i$ during the algorithm's runtime. Such properties may be for example - the size of the maximum matching in $K_i$, the number of edges in some subgraph, and etc. The idea behind this is that such algorithms will be easier to analyze because we may assume the intermediate states $K_i$ satisfied several properties.

Each constraint will be associated with some parameter which receives values out of some finite domain. The (ordered) set of all parameters will be referred to as the template, and each proper assignment of the template will be referred to as a signature. Let $\mathcal{S}$ be the set of all possible signatures (note that $\mathcal{S}$ may be a subset of the Cartesian product of the parameters' domains). We say that an algorithm $\mathcal{A}$ had a signature $p\in\mathcal{S}$ if the values of all parameters in the template were according to $p$ during the runtime of the algorithm. For example, we may use two constraints - the size of the maximum matching in $K_1$, and the number of edges in $K_1$. The template will be $(m_1,e_1)$, and the set of signatures will be some subset {of the following set (considering only pairs of integers): 
\[\mathcal{S}\subseteq\Big[0,\frac{k}{2}\Big]\times\Big[0,\kct\Big]\]
}

For a signature $p\in\mathcal{S}$, and an algorithm $\mathcal{A}$, we define the algorithm $\mathcal{A}_p$ as follows. $\mathcal{A}_p$ succeeds, and returns the clique $K$, if and only if $\mathcal{A}$ returned $K$ and had the signature $p$. Otherwise, the algorithm fails. In the example with template $(m_1,e_1)$, if the maximum matching of $K_1$ when $\mathcal{A}$ runs is $a^*$ and the number of edges is $b^*$, then only $\mathcal{A}_p$ with $p=(a^*,b^*)$ returns $K$ (assuming $\mathcal{A}$ succeeds).
Note that if $\mathcal{A}$ succeeds, there is exactly one $p\in \mathcal{S} $ for which the algorithm $\mathcal{A}_p$ succeeds. If $\mathcal{A}$ fails, all $\mathcal{A}_p$'s fail too. 

Due to the extra structure, it will be simpler to analyze the algorithms $\mathcal{A}_p$ instead of $\mathcal{A}$. Given some upper bound on the probability $\mathcal{A}_p$ succeeds for any $p$, we will derive a bound on the probability $\mathcal{A}$ succeeds, and then derive a bound on $\al_*(\de,l)$. This can be generalized to the case where we have several upper bounds (denote their amount by $B$). We will assume that the $B$ bounds have the following form:
\[\mathbb{P}(\mathcal{A}_p \text{ succeeds})\leq n^{f_1^j(p)k}\cdot 2^{-f_2^j(p)\kct}\]
where $f_1^j,f_2^j:\mathcal{S}\rightarrow \mathbb{R}$ for every $1\leq j\leq B$, and $f_1^j(p),f_2^j(p)$ do not depend on $n$. We give a brief overview to the reason we consider bounds of this form, which will become more clear when we prove the actual bounds.
\begin{itemize}
    \item $n^{f_1^j(p)k}$ will be the number of candidate sets (sets of vertices of $G$) that satisfy the signature $p$, and possibly may be returned as the clique $K$.
    \item $2^{-f_2^j(p)\kct}$ will be the probability that some subset of the edges within a candidate set will be positive. If this event occurs, we refer to the candidate set as a feasible set. A clique $K$ returned by the algorithm must be a feasible set.
    \item The product $n^{f_1^j(p)k}\cdot 2^{-f_2^j(p)\kct}$ will be an upper bound on the expected number of feasible sets. Using this bound we will derive a bound on $\mathbb{P}(\mathcal{A}_p \text{ succeeds})$ (using Markov's inequality).
\end{itemize}
In order to prove upper bounds on $\al_*(\de,l)$ using upper bounds on $\mathbb{P}(\mathcal{A}_p \text{ succeeds})$, we prove the following lemma. 
\begin{lemma}
\label{cons_lemma}
Let $1\leq \de < 2$, $l\geq 1$, and let $\mathcal{S}$ be a set of signatures such that $|\mathcal{S}|=n^{o(\log n)}$. Suppose we have several upper bounds (denote their amount by $B$) on the probability algorithm $\mathcal{A}_p$ succeeds. Suppose the bounds are of the following form: $\mathbb{P}(\mathcal{A}_p \text{ succeeds})\leq n^{f_1^j(p)k}\cdot 2^{-f_2^j(p)\kct}$ where $f_1^j,f_2^j:\mathcal{S}\rightarrow \mathbb{R}$ for every $1\leq j\leq B$, and $f_1^j(p),f_2^j(p)$ do not depend on $n$. Then:
\[\al_*(\de,l)\leq \max_{p\in \mathcal{S}}\min_{j}\frac{2f_1^j(p)}{f_2^j(p)}\]
\end{lemma}

\begin{proof}
Denote $\al= \al_*(\de,l)$, $k=\al\log n$, and let $\mathcal{A}$ be a $(\de,l,\al)$-algorithm. Assume for contradiction $\al>\max_{p\in \mathcal{S}}\min_{j}\frac{2f_1^j(p)}{f_2^j(p)}$. Given the upper bounds, we have that $\mathbb{P}(\mathcal{A}_p \text{ succeeds})\leq \min_j n^{f_1^j(p)k}\cdot 2^{-f_2^j(p)\kct}$. If $\mathcal{A}$ succeeds, there must be some $p\in\mathcal{S}$ for which $\mathcal{A}_p$ succeeds (we define $\mathcal{S}$ to include all possible outcomes). Hence, using union bound we get:
\begin{align*}
    \frac{1}{2}&\leq
\mathbb{P}(\mathcal{A}\text{ succeeds}) \\
&\leq\sum_{p\in \mathcal{S}} \mathbb{P}(\mathcal{A}_p\text{ succeeds}) \\
&\leq\Big[\max_{p\in \mathcal{S}} \mathbb{P}(\mathcal{A}_p\text{ succeeds})\Big]\cdot n^{o(\log n)} \\
&\leq\Big[\max_{p\in \mathcal{S}} \min_j \Big(n^{f_1^j(p)k}\cdot 2^{-f_2^j(p)\kct}\Big)\Big]\cdot n^{o(\log n)}
\end{align*}

Taking $\log$ of both sides we get:
\begin{align*}
    \log \Big(\Big[\max_{p\in \mathcal{S}}\min_j \Big(n^{f_1^j(p)k}\cdot 2^{-f_2^j(p)\kct}\Big)\Big]\cdot n^{o(\log n)}\Big) \geq& -1 \\
    \max_{p\in \mathcal{S}} \min_j \Big[{f_1^j(p)k}\log n -f_2^j(p)\kct\Big]+ o(\log^2 n) \geq& -1 \\
    \max_{p\in \mathcal{S}} \min_j\Big[(2f_1^j(p)-\al f_2^j(p))\Big]\frac{1}{2}\al \log^2 n+ o(\log^2 n) \geq& -1 
\end{align*}
For sufficiently large $n$, this expression is dominated by the term $\log^2 n$. Since we assumed that $\al> \max_{p\in \mathcal{S}}\min_j \frac{2f_1^j(p)}{f_2^j(p)}$, we have that for every $p\in \mathcal{S}$ there exists $j$ such that:
\begin{gather*}
    \al>\frac{2f_1^j(p)}{f_2^j(p)} \\
    \implies 2f_1^j(p)-\al f_2^j(p)<0
\end{gather*}
\[\]
Therefore, $ \max_{p\in \mathcal{S}} \min_j\Big[(2f_1^j(p)-\al f_2^j(p))\Big]<0$ and as $n\rightarrow\infty$ we have that:
\[\max_{p\in \mathcal{S}} \min_j\Big[(2f_1^j(p)-\al f_2^j(p))\Big]\frac{1}{2}\al \log^2 n+ o(\log^2 n)\rightarrow -\infty\]
Which is a contradiction.
\end{proof}


We show examples of the use of Lemma \ref{cons_lemma} in the proofs of lemmas \ref{one_round_lemma} and \ref{gen_lemma}. Lemma \ref{one_round_lemma} was proved in \cite{feige2020finding}. Here we illustrate how Lemma \ref{cons_lemma} can be used in order to prove this lemma.

\begin{lemma}
\label{one_round_lemma}
For every $1\leq \de<2$:
\[\al_*(\de,1)\leq \de\]
\end{lemma}
\begin{proof}
We will use Lemma \ref{cons_lemma} with a template which consists of a single parameter, the size of the maximum matching in $K_1$. An algorithm $\mathcal{A}_p$ may succeed only if the maximum matching in $K_1$ is a perfect matching, i.e., of size $\frac{1}{2}k$. Therefore, there is only one relevant signature of the template to consider, which we denote $p_0$. In order to bound $\mathbb{P}(\mathcal{A}_{p_0}\text{ succeeds})$, we first bound the number of subgraphs that may be the returned clique $K$. A set of vertices $a$ will be called a candidate if it satisfies the following:
\begin{itemize}
    \item $|a|=k$.
    \item $Q_1$ induced on $a$ has a perfect matching of size $\frac{1}{2}k$.
\end{itemize}
Note that if the algorithm succeeds, it must return a candidate. The number of candidates is bounded by the number of ways to choose a perfect matching. The perfect matching is a set of $\frac{1}{2}k$ edges out of the $n^\de$ edges in $Q_1$. Hence there are at most ${{n^\de}\choose{\frac{1}{2}k}}\leq n^{\frac{1}{2}\de k}$ candidates. For a candidate to be a clique, all $\kct$ edges within it must be positive, and this happens with probability $2^{-\kct}$. Therefore, $\mathbb{P}(\mathcal{A}_{p_0}\text{ succeeds})\leq n^{\frac{1}{2}\de k}\cdot 2^{-\kct}$, and by Lemma \ref{cons_lemma}:
\[\al_*(\de,1)\leq \de\]
\end{proof}
For the next lemma, we use a template which consists of the following parameters:
\begin{itemize}
    \item $m_i\in[0,\frac{1}{2}]$ where $1\leq i\leq l$ and $m_i k$ is the size of the maximum matching in $K_i$.
    \item $w_i\in[0,1]$ where $0\leq i\leq l$ and $w_i\kct$ is the number of edges of $K$ not queried up to round $i$ (i.e., the number of edges in the complement of $K_i$). Note that $w_0=1$.
\end{itemize}

\begin{lemma}
\label{gen_lemma}
For every $l\in\mathbb{N}$, and $1\leq\de< 2$:
\[\al_*(\de,l)\leq \max_{p\in\mathcal{S}}\min_{1\leq i\leq l}\frac{2-(4-2\de)m_i}{w_{i-1}}\]
\end{lemma}

\begin{proof}
Let $1\leq i\leq l$, and let $p\in \mathcal{S}$ be some signature. We start with an upper bound on the number of subgraphs of $G$ which after round $i$ may be the subgraph $K_i$. Such subgraphs, which will be called $i$-candidates, must satisfy several conditions in order to comply with the signature. A set of vertices $a$ will be called an $i$-candidate if it satisfies the following:
\begin{itemize}
    \item $|a|=k$.
    \item $Q_{\leq i}$ induced on $a$ has a maximum matching of size $m_ik$.
    \item {$Q_{\leq (i-1)}$} induced on $a$ has $(1-w_{i-1})\kct$ edges.
\end{itemize}
{An $i$-candidate $a$ may be partitioned to $b\sqcup c$ where $b$ consists of the $2m_i k$ vertices covered by a maximum matching (breaking ties arbitrarily), and $c$ is the rest of the vertices ($\abs{c}=(1-2m_i)k$).
The maximum matching must be some subset of $m_ik$ edges among the $in^{\de}$ edges queried up to round $i$, so we can upper bound the number of different matchings by the number of subsets of $m_ik$ queries. Hence, there are at most ${{in^{\de}}\choose{m_ik}}$ ways to choose the set $b$. The $(1-2m_i)k$ vertices in $c$ may be any vertices of {$V$}, so there are at most ${{n}\choose{k(1-2m_i)}}$ ways to choose $c$.} Therefore, the number of $i$-candidates is at most ${{in^{\de}}\choose{m_ik}}{{n}\choose{k(1-2m_i)}}\leq n^{\de m_ik+k(1-2m_i)}=n^{(1-(2-\de)m_i)k}$.

The queries up to round $i$ are determined by the answers to the queries up to round $i-1$. The answers given on round $i$ {and on later rounds,} do not affect which subgraphs belong to the family of $i$-candidates (as opposed to answers given up to round $i-1$). Any pair of vertices not queried up to round $i-1$ has probability $\frac{1}{2}$ of being answered positively when queried, and in any $i$-candidate there are $w_{i-1}\kct$ such potential edges. If all these edges in a given $i$-candidate are answered positively, we say that the candidate is $i$-feasible. The probability an $i$-candidate is $i$-feasible is $2^{-w_{i-1}\kct}$.

The clique $K$ returned by the algorithm must be an $i$-feasible set for every $i$. Hence, in order for $\mathcal{A}_p$ to succeed, there must be a set which is $i$-feasible for every $i$. Using union bound we get that the probability $\mathcal{A}_p$ will succeed is at most $n^{(1-(2-\de)m_i)k}\cdot 2^{-w_{i-1}\kct}$ for every $i$. In addition, we have $|\mathcal{S}|=polylog(n)$. Therefore, by Lemma \ref{cons_lemma}:
\[\al_*(\de,l)\leq \max_{p\in\mathcal{S}}\min_{1\leq i\leq l}\frac{2-(4-2\de)m_i}{w_{i-1}}\]
\end{proof}

In \cite{feige2020finding} a variation of Lemma \ref{gen_lemma} was used, where the constrained parameters were $w_i$, and bounds on $m_i$ were deduced from them. We will use the opposite direction, which is simpler, and deduce bounds on $w_i$ from $m_i$.

\subsubsection{Restricting the signatures set}

When we use Lemma \ref{cons_lemma} it will sometimes be useful to restrict the signatures set over which we maximize, when we bound $\al_*(\de,l)$. For this purpose, we include the size of the maximum independent set in $K_{l-1}$ as a parameter in the template, denoted $I_{l-1}$. An independent set in $K_{l-1}$ must be a clique fully queried in round $l$, in order for $K_l$ to be a clique. Finding any clique of size $I_{l-1}$ in round $l$, considering only subgraphs completely un-queried up to round $l-1$, is at least as hard as finding a clique of size $I_{l-1}$ using a one round algorithm (since all queries are answered independently). Using Lemma \ref{one_round_lemma}, we cannot expect to find a clique if $I_{l-1}> \de\log n$. This is formulated in the following lemma, which is proved in Appendix \ref{appendix_proof_max_indep}.

\begin{lemma}
\label{max_indep_lemma}
    Let $\mathcal{S}_{I_{l-1}\leq \de\log n}\sqcup \mathcal{S}_{I_{l-1}> \de\log n}$ be a partition of the signatures in $\mathcal{S}$ according to the size of $I_{l-1}$ in relation to $\de\log n$. Lemma \ref{cons_lemma} holds when $\mathcal{S}$ is restricted to $\mathcal{S}_{I_{l-1}\leq \de\log n}$, i.e.,
    \[\al_*(\de,l)\leq \max_{p\in \mathcal{S}_{I_{l-1}\leq \de\log n}}\min_{j}\frac{2f_1^j(p)}{f_2^j(p)}\]
\end{lemma}

\subsection{Gallay-Edmonds decomposition (GED)}
\label{sec_ged}
In order to prove upper bounds in the case of three rounds, we will use the Gallay-Edmonds Decomposition (see, e.g., \cite{lovasz2009matching}) of the subgraphs $K_i$. A graph $H$, with maximum matching of size $m$, decomposes into three sets (see Figure \ref{GED_fig}). The set $C$ consists of all vertices such that there is a maximum matching in which they are not covered. The set $S$ (separator) consists of the neighbors of $C$  ($S=N(C)\setminus C$), and $R$ is the rest of the vertices (i.e., non-neighbors of $C$). The decomposition satisfies the following: 
\begin{itemize}
    \item The set $C$ is composed of $c$ connected components of odd size (referred to as odd components) and no components of even size. The size of $R$ is even.
    \item Any maximum matching satisfies:
    \begin{itemize}
        \item $R$ is matched to itself (perfectly).
        \item in any odd component, all vertices but one are matched to each other (perfectly).
        \item All vertices of $S$ are matched to vertices of $C$ of distinct odd components (to vertices not matched within $C$).
    \end{itemize}
    \item In any maximum matching there are $c-|S|$ unmatched vertices, so $c-|S|=|V(H)|-2m$.
\end{itemize}
Given a specific maximum matching $M$, we will sometimes use another definition of the decomposition, which depends on $M$. We refer to it as a specific matching decomposition. It is composed of the sets $C^{*+},C^{*-},S,D$ defined as follows (see Figure \ref{GED_fig}):
\begin{itemize}
    \item $S$ is the same as in Gallay-Edmonds Decomposition (GED).
    \item $C^{*+}$ is the set of vertices matched to $S$.
    \item $C^{*-}$ is the set of unmatched vertices. Denoting $C^*=C^{*+}\cup C^{*-}$ we have $|C^{*}|=c$.
    \item $D$ is the set of the rest of the vertices(i.e., $R$ and the even matched part of each odd component).
\end{itemize}
Note that $C^*$ is an independent set and that the number of edges between $D$ and $C^*$ is at most linear in $|V(H)|$ (since any vertex in $D$ has at most one neighbor in $C^*$). 

\begin{figure}[h!]
  \centering
  \includegraphics[width=1\textwidth]{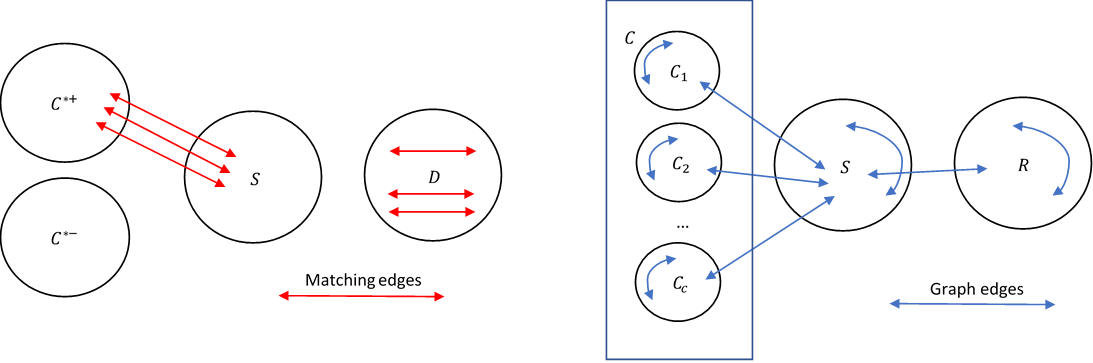}
  \caption{Gallay-Edmonds decomposition (right) and specific matching decomposition (left)}
  \label{GED_fig}
\end{figure}

\section{Upper Bounds}
\subsection{Two rounds}
\subsubsection{Tight bound for two rounds with small $\de$}

We prove an upper bound on $\al_*(\de,2)$ for every $1 \leq \de <2$, which is tight when $1\leq \de\leq \frac{6}{5}$. 
\begin{theorem}
\label{two_rounds_small_delta_thm}
Let $1\leq\de< 2$. Then:
\[\al_*(\de,2)\leq\frac{4}{3}\de\]
\end{theorem}

\begin{proof}
Denote by $M_1$ some matching of maximum size in the subgraph of $Q_1$ induced on $V(K)$, i.e., in $K_1$ (breaking ties arbitrarily). Its size is denoted by $m_1 k$, where $0 \le m_1 \le \frac{1}{2}$. The set of $2m_1k$ vertices covered by $M_1$ is denoted by $V_{M_1}$. The number of edges in the subgraph of $Q_1$ induced on $V_{M_1}$ is denoted by $e_1 {k \choose 2}$, where $0\leq e_1 \leq 1$. Both $m_1,e_1$ will be in the template we will use. 

Consider $K_1$ which satisfies the signature $(m_1,e_1)$, and the corresponding set $V_{M_1}$. The set $K_1\setminus V_{M_1}$ is independent, otherwise we could add an edge to the matching $M_1$, increasing its size. Hence, for $K$ to be a clique it is necessary that there will be a perfect matching in the subgraph of $Q_2$ induced on $V(K)\setminus V_{M_1}$. Denote such a matching by $M_2$ (breaking ties arbitrarily). Its size is denoted by $m_2 k$, hence $m_1 + m_2 = \frac{1}{2}$. Denote $V_{M_2}=V(K)\setminus V_{M_1}$, and denote the number of edges in the subgraph of $Q_2$ induced on $V(K)$ by $e_2 {k \choose 2}$. Observe that $e_1 + e_2 \le 1$, but equality need not hold, as some pairs $u,v$ for which $v \in V_{M_1}$ and $u \in V_{M_2}$, may have been queried in the first round. We add the parameter $e_2$ to the template.

For every signature $p=(m_1,e_1,e_2)$ we wish to bound the probability that there is a clique $K$ with signature $p$. For such $K$, $V(K)$ must decompose into $V_{M_1}, V_{M_2}$, and all pairs of vertices in $K$ must be queried and answered positively, particularly the indicated $(e_1+e_2)\kct$ pairs. 

A set of vertices $a$ will be a $1$-candidate if it is a possible $V_{M_1}$ for some $K$, i.e., satisfies the following:
\begin{itemize}
    \item $|a|=2m_1k$.
    \item $Q_1$ induced on $a$ has a perfect matching.
    \item $Q_1$ induced on $a$ has $e_1\kct$ edges.
\end{itemize}
The property of a set being a $1$-candidate depends only on $Q_1$, and hence determined deterministically. There are at most ${q_1}\choose {m_1k}$ ways to choose the matching $M_1$ which defines $V_{M_1}$, hence at most ${q_1\choose m_1k}$ $1$-candidates.

$a$ will be called $1$-feasible if:
\begin{itemize}
    \item $a$ is a $1$-candidate.
    \item All $e_1\kct$ queries of $Q_1$ induced on $a$ were positive.
\end{itemize}
The probability a $1$-candidate is feasible is $2^{-e_1\kct}$, so in expectation there are at most ${q_1\choose m_1k}\cdot 2^{-e_1\kct}$ $1$-feasible sets.

Two sets $a,b$ will be a candidate pair if $(a,b)$ are a possible pair for $V_{M_1},V_{M_2}$ of some $K$, i.e., satisfy the following:
\begin{itemize}
    \item $a$ is $1$-feasible.
    \item $a\cap b=\emptyset$.
    \item $|b|=2m_2k$.
    \item $Q_1$ induced on $b$ is an independent set.
    \item $Q_2$ induced on $b$ has a perfect matching.
    \item $Q_2$ induced on $a\cup b$ has $e_2\kct$ edges.
\end{itemize}
A pair $(a,b)$ will be feasible if: 
\begin{itemize}
    \item $(a,b)$ is a candidate pair.
    \item All $e_2\kct$ queries of $Q_2$ induced on $a\cup b$ were positive.
\end{itemize}

Given $a$, the sets $b$ which make a candidate pair with $a$ 
{depend on $Q_2$, and $Q_2$ depends on the answers to $Q_1$. 
For every choice of answers to $Q_1$,} there are at most ${q_2}\choose {m_2k}$ ways to choose the matching $M_2$ which defines $V_{M_2}$, so this bounds the number of possible such sets $b$. 
Round $2$ queries are answered positively with probability $\frac{1}{2}$ independently of $Q_1$ answers. In addition, the number of relevant pairs, $e_2\kct$, is chosen in advance. Hence, the probability a candidate pair will be feasible is $2^{-e_2\kct}$, and in expectation there are at most ${q_2\choose m_2k}\cdot 2^{-e_2\kct}$ feasible pairs $(a,b)$ for any $1$-feasible set $a$. 

The probability $\mathcal{A}_p$ succeeds (where $p=(m_1,e_1,e_2)$) is at most the probability a feasible pair exists, which is at most the expected number of such pairs (by Markov's inequality). By linearity of expectation, we can write this expectation as the sum of expectations conditioned on a specific $1$-feasible set $a$. Hence:
\[\mathbb{P}(\mathcal{A}_p\text{ succeeds})\leq {q_1\choose m_1k}\cdot 2^{-e_1\kct}\cdot {q_2\choose m_2k}\cdot 2^{-e_2\kct}\]
Given $q_1=q_2=n^\de$ we get:
\[\mathbb{P}(\mathcal{A}_p\text{ succeeds})\leq n^{\frac{1}{2}\de k}\cdot 2^{-(e_1+e_2)\kct}\]
{
The edges in $K$ may be partitioned into three types:
\begin{enumerate}
    \item $e_1\kct$ edges within $V_{M_1}$, queried in round 1.
    \item $e_2\kct$ edges within $V(K)$, queried in round 2.
    \item $(1-e_1-e_2)\kct$ edges between $V_{M_1}$ and $V_{M_2}$, queried in round 1. Note that there are no edges within $V_{M_2}$ queried in round 1.
\end{enumerate}
After receiving $Q_1$ answers, the algorithm might choose $Q_2$ such that all candidate pairs $a,b$ have only positive $Q_1$ queries between them. So the algorithm may choose to include edges in candidates for $K$, based on previous information that they are positive. Hence we cannot consider those edges probability of being positive when bounding the probability a candidate for $K$ is a clique. We therefore refer to such edges (of the third type) as "free edges" for the algorithm. This clarifies why we defined feasibility of a candidate pair to depend only on the $(e_1+e_2)\kct$ edges of the first two types, which are answered after the algorithm chose to include them in a candidate.}

The following lemma bounds the number of free edges. Its proof appears after the end of the proof of the current theorem.

\begin{lemma}
\label{matching_edges_half}
Let $H$ be a graph with an even number of vertices, let $V_M$ be the set of vertices matched by a maximum matching in $H$, and let $V_M^-$ be the rest of the vertices. Denote by $E^-$ the set of edges of $H$ not contained in $V_M$ (i.e., have at least one endpoint in $V_M^-$). Then $E^-= E(V_M,V_M^-)$ and
$|E^-|\leq \frac{1}{2}|V_M||V_M^-|$.
\end{lemma}

Applying Lemma \ref{matching_edges_half} with $H=K_1$ we get $(1-e_1-e_2)\kct\leq \frac{1}{2}\cdot 2m_1k\cdot 2m_2k$ and then (using $\kct\simeq \frac{k^2}{2}$, neglecting low order terms):
\[e_1+e_2\geq 1-4m_1m_2=1-4m_1(\frac{1}{2}-m_1)\]
Therefore:
\[\mathbb{P}(\mathcal{A}_p\text{ succeeds})\leq n^{\frac{1}{2}\de k}\cdot 2^{-(1-4m_1(\frac{1}{2}-m_1))\kct}\]

We apply Lemma \ref{cons_lemma} with the signatures set $\mathcal{S}$ defined by $m_1,e_1,e_2$ which satisfies $\mathcal{S}=O(\log^5 n)$. We get:
\[\al_*(\de,2)\leq \max_{m_1}\frac{\de}{1-4m_1(\frac{1}{2}-m_1)}\]
This expression is maximal when $m_1(\frac{1}{2}-m_1)$ is maximal, i.e., for $m_1=\frac{1}{4}$, and we get the desired upper bound:
\[\al_*(\de,2)\leq\frac{4}{3}\de\]

\end{proof}

\begin{corollary}
For $1\leq \de\leq \frac{6}{5}$ the upper bound is tight and we have:
\[\al_*(\de,2)=\frac{4}{3}\de\]
\end{corollary}
The signature which maximizes the objective function has the following setting of parameters: $m_1=\frac{1}{4}$ and $e_1,e_2$ are such that $\frac{1}{2}|V_{M_1}||V_{M_2}|$ edges are queried between $V_{M_1}$ and $V_{M_2}$ on round $1$. These values are equal to those in the signature of Algorithm 2 (see Section \ref{algos_sec}), as we could expect would happen in an optimal algorithm. We finish with a proof of Lemma \ref{matching_edges_half}.

\begin{proof}
The set $V_M^-$ must be independent, because if it contains some edge we could add the edge to the matching, contradicting its maximality. Hence, $E^-= E(V_M,V_M^-)$. Let $v,u\in V_M$ be a pair matched by the maximum matching. Assume without loss of generality (w.l.o.g) that $\deg^-(u)\leq \deg^-(v)$, where $\deg^-$ is the number of neighbors in $V_M^-$ the vertex has. We claim that $\deg^-(u)\geq 1$ implies $\deg^-(u)=\deg^-(v)=1$. Assuming $\deg^-(u)\geq 1$, Let $w\in V_M^-$ be a neighbor of $u$, and assume for contradiction $\deg^-(v)\geq 2$. Then, there is some $w\neq z\in V_M^-$ which is a neighbor of $v$. Replacing the edge $\{u,v\}$ in the matching by $\{u,w\}$ and $\{v,z\}$ we get a larger matching, which is a contradiction. Hence, either $\deg^-(u)=0$ or $\deg^-(u)=\deg^-(v)=1$. In both cases, since $|V_M^-|\neq 1$ (as $|V(H)|$ is even), this implies $\deg^-(u)+\deg^-(v)\leq |V_M^-|$. Since $E^-= E(V_M,V_M^-)$ we have:
\[|E^-|=\sum_{u\in V_M} \deg^-(u)=\sum_{u\text{ matched to }v } \deg^-(u)+\deg^-(v)\leq \frac{1}{2}|V_M||V_M^-|\]
\end{proof}

\subsubsection{Bound for two rounds with large $\de$}
\label{sec_two_rounds_large_delta}
The upper bound $\al_*(\de,2)\leq\frac{4}{3}\de$ proved in Theorem \ref{two_rounds_small_delta_thm} is tight only for $1\leq \de\leq \frac{6}{5}$. For $\frac{6}{5}< \de<2$ we had a lower bound of $1+\frac{\de}{2}$, which is smaller on this range of $\de$ (see Figure \ref{bounds_fig}). Resolving the two rounds with $\frac{6}{5}< \de<2$ case is important before proceeding to the analysis of more rounds with any value of $\de$. The reason is that we may use an algorithm for two rounds with large $\delta$, to solve the problem of $l>2$ rounds with small $\delta$. For example we prove the following lemma.
\begin{lemma}
If there exists $\frac{6}{5}< \de<2$ and $\epsilon>0$ for which $\al_*(\de,2)\geq 1+\frac{\de}{2}+\epsilon$, then for some $\epsilon'>0$, $\al_*(1,4)\geq \frac{3}{2}+\epsilon'$. (Note that the current lower bound is $\al_*(1,4)\geq\frac{3}{2}$).
\end{lemma}

\begin{proof}
Suppose that $\al_*(\de,2)\geq 1+\frac{\de}{2}+\epsilon$, and let $\mathcal{A}^2$ be a $(\de,2,1+\frac{\de}{2}+\epsilon)$-algorithm. We present a $(1,4,\frac{3}{2}+\epsilon')$-algorithm $\mathcal{A}^4$. In round 1, $\mathcal{A}^4$ queries all pairs in a subset $S$ of size $\sqrt{n}$, and finds a clique $S'$ of size $\beta \log n$, for $\beta = 1 - \frac{1}{\delta} < 1$. In round 2, $\mathcal{A}^4$ queries all pairs between $S'$ and a disjoint set $T$ of size $\frac{n}{|S'|}$, and finds the set of mutual neighbors $T'$. $T'$ is of expected size $n^{1-\beta}$ (neglecting low order terms). Then, in rounds 3 and 4, algorithm {$\mathcal{A}^4$ uses algorithm  $\mathcal{A}^2$ on the subgraph $T'\subset G$} and finds a clique $T''\subset T'$. Specifically, denoting $N=|T'|=n^{1-\beta}$ and 
noting that $\de=\frac{1}{1-\beta}$, we have that the number of queries in rounds 3 and 4 is $N^\de=n$. Hence, using $\mathcal{A}^2$ on $T'$ with $N^\de$ queries is feasible. By the assumption, the clique $T''$ found this way is of size at least 
\[\Big(1+\frac{\de}{2}+\epsilon\Big)\log N=\Big(1+\frac{1}{2(1-\beta)}+\epsilon\Big)(1-\beta)\log n=\Big(\frac{3}{2}-\beta+(1-\beta)\epsilon\Big)\log n\]
Algorithm $\mathcal{A}^4$ outputs the clique $K\coloneqq S'\cup T''$, which is of size at least $(\frac{3}{2}+\epsilon')\log n$, where $\epsilon'\coloneqq (1-\beta)\epsilon$.
Hence $\al_*(\de,2)\geq 1+\frac{\de}{2}+\epsilon$ implies $\al_*(1,4)\geq \frac{3}{2}+\epsilon'$.
\end{proof}


In the following theorem we improve the upper bound $\al_*(\de,2)\leq\frac{4}{3}\de$ when $\frac{6}{5}<\de<2$.
\begin{theorem}
\label{two_rounds_big_delta_full_thm}
For every $\frac{6}{5}<\de<2$:
\[\al_*(\de,2)\leq 1+\sqrt{(\de-1)(3-\de)}\]
\end{theorem}
\begin{proof}
    In the proof of Theorem \ref{two_rounds_small_delta_thm} we prove that for every $1\leq\de<2$:
    \[\al_*(\de,2)\leq \max_{m_1}\frac{\de}{1-4m_1(\frac{1}{2}-m_1)}\]
    Denote $\al_2^{\de}(m_1)\coloneqq \frac{\de}{1-4m_1(\frac{1}{2}-m_1)}$.

    {We combine this upper bound with Lemma \ref{gen_lemma}. For convenience, we replace the notation of Lemma \ref{gen_lemma} with the notation defined in Theorem \ref{two_rounds_small_delta_thm}.
    
    \begin{align*}
        \al_*(\de,2)&\leq  \max_{p\in\mathcal{S}}\min\Big\{{2-(4-2\de)m_1}, \frac{2-(4-2\de)(m_1+m_2)}{e_{2}}\Big\}=\\
        &=\max_{p\in\mathcal{S}}\min\Big\{{2-(4-2\de)m_1}, \frac{\de}{e_2}\Big\}
    \end{align*}
    
    Denote $\al_1^{\de}(m_1)\coloneqq 2-(4-2\de)m_1$. We combine both upper bounds by applying Lemma \ref{cons_lemma} on the corresponding upper bounds on $\mathbb{P}(\mathcal{A}_p \text{ succeeds})$. The expression $\frac{\de}{e_2}$ imposes a weaker bound than the one proved in Theorem \ref{two_rounds_small_delta_thm}, so we omit it and get the following upper bound:
    
    \[\al_*(\de,2)\leq \max_{m_1}\min\{ \al_1^{\de}(m_1),\al_2^{\de}(m_1)\}\]
    
    }
    The maximum over $m_1\in[0,\frac{1}{2}]$ may be obtained either on the boundary ($m_1\in\{0,\frac{1}{2}\}$), on a local maximum of $\al_1^{\de}$ or $\al_2^{\de}$, or when $\al_1^{\de}(m_1)=\al_2^{\de}(m_1)$. For $m_1\in\{0,\frac{1}{2}\}$ we have $\min\{ \al_1^{\de}(m_1),\al_2^{\de}(m_1)\}=\de$ which is not maximal. The only local maximum is in $\al_2^{\de}$ when $m_1=\frac{1}{4}$. For this $m_1$ we have $\al_2^{\de}(\frac{1}{4})=\frac{4}{3}\de$ and $\al_1^{\de}(\frac{1}{4})=1+\frac{\de}{2}$. Therefore, for any $\frac{6}{5}<\de<2$ we have $\al_2^{\de}(\frac{1}{4})>\al_1^{\de}(\frac{1}{4})$, so the maximum of $\min\{ \al_1^{\de}(m_1),\al_2^{\de}(m_1)\}$ cannot be obtained in $m_1=\frac{1}{4}$. Hence, the maximum is obtained when $\al_1^{\de}(m_1)=\al_2^{\de}(m_1)$. Solving this equation we get:
    \begin{gather*}
        \frac{\de}{1-4m_1(\frac{1}{2}-m_1)}=2-(4-2\de)m_1 \\
        \de = 2-4m_1+8m_1^2-(4-2\de)m_1+(8-4\de)m_1^2-(16-8\de)m_1^3 \\
        0=-(16-8\de)m_1^3+(16-4\de)m_1^2-(8-2\de)m_1+2-\de
    \end{gather*}
    The only solution of the cubic equation satisfying $m_1\in(0,\frac{1}{2})$ is:
    \[m_1=\frac{-\sq{(\de-1)(3-\de)}+1}{4-2\de}\]
    For this value of $m_1$ we have $\al_1^{\de}(m_1)=\al_2^{\de}(m_1)=1+\sq{(\de-1)(3-\de)}$, hence:
    \[\al_*(\de,2)\leq 1+\sqrt{(\de-1)(3-\de)}\]
\end{proof}
\subsubsection{Tight bound for two rounds with large $\de$, restricted first round}
\label{sec_two_rounds_large_delta_restricted}

{The upper bound proved in Theorem \ref{two_rounds_big_delta_full_thm} may not be tight, as the current best lower bound for $\frac{6}{5}<\de<2$ is $\al_*(\de,2)\geq 1+\frac{\de}{2}$. This lower bound is achieved by Algorithm 3, which finds a clique of size $(1+\frac{\de}{2})\log n$. The gap between the bounds might be due to the fact that Algorithm 3 queries less than $n^\de$ pairs in round 1, while in Theorem \ref{two_rounds_big_delta_full_thm} we prove an upper bound for algorithms that can query $n^\de$ pairs in round 1. Therefore, we raise the question whether Algorithm 3 is optimal among all algorithms which query the same number of queries per round as Algorithm 3 does (which means the number of round 1 queries in those algorithms is restricted). In round $1$, Algorithm 3 queries all pairs between a set of size $n^{\frac{1}{2}-\frac{\de}{4}}$ and the rest of the graph. Hence, the algorithm makes roughly $n^{\frac{3}{2}-\frac{\de}{4}}$ queries in round $1$. In round 2, the algorithm queries all pairs within a set of size $n^{\de/2}$, which is roughly $n^\de$ queries. For $\frac{6}{5}<\de<2$ Denote $\de_1=\frac{3}{2}-\frac{\de}{4}$. We will upper bound the maximum size of a clique that is likely to be found using an algorithm that queries $n^{\de_1}$ pairs in round 1, and $n^{\de}$ pairs in round 2. This is formulated in the following definitions:}

\begin{definition}
\label{def_rest_dif_1}
     Let $l\in\mathbb{N}$, $0<\de_i<2$ for every $1\leq i\leq l$, and let $\al\leq 2$. A deterministic algorithm $\mathcal{A}$, which makes $l$ adaptive round queries with at most $n^{\de_i}$ queries in each round, will be called a $((\de_1,\dots,\de_l),l,\al)$-algorithm
     {if for any sufficiently large $n$, with probability at least $\frac{1}{2}$, $G'$ contains a clique of size $k\ge \al\log  - o(\log n)$.} 
\end{definition}
\begin{definition}
\label{def_rest_dif_2}
    Let $l\in\mathbb{N}$, and let $0<\de_i<2$ for every $1\leq i\leq l$. $\al_{*}((\de_1,\dots,\de_l),l)$ is defined as the supremum over $\al$ such that there exists a $((\de_1,\dots,\de_l),l,\al)$-algorithm.
\end{definition}

{Due to the new definitions, some adjustments are needed in previously proved lemmas.} The proof of Lemma \ref{cons_lemma} depends on $\de$ only implicitly, so the same proof holds also using definitions \ref{def_rest_dif_1} and \ref{def_rest_dif_2}. Lemma \ref{gen_lemma} requires some adjustments.

\begin{lemma}
\label{gen_lemma2}
Using the same notation of $m_i,w_i$ as in Lemma \ref{gen_lemma}, for every $l\in\mathbb{N}$, and $0<\de_i< 2$:
\[\al_*((\de_1,\dots,\de_l),l)\leq \max_{p\in\mathcal{S}}\min_{1\leq i\leq l}\frac{2-(4-2 \max_{j\leq i}{\de_j})m_i}{w_{i-1}}\]
\end{lemma}
\begin{proof}
The only difference between this case and the previous one is that the matching of size $m_ik$ is chosen out of $\sum_{j\leq i}n^{\de_j}\leq i\cdot n^{\max_{j\leq i}\de_j}$ queries, and not $in^{\de}$. The set of vertices not covered by the matching has the same number of candidates and the probability analysis is the same. Hence, we get $\mathbb{P}(\mathcal{A}_p\text{ succeeds})\leq n^{(1-(2- \max_{j\leq i}{\de_j})m_i)k}\cdot 2^{-w_{i-1}\kct}$. We get the desired bound by Lemma \ref{cons_lemma}.

\end{proof}

In the following theorem we show that Algorithm 3 is optimal when the first round is restricted.
\begin{theorem}
\label{two_rounds_big_delta_restricted_thm}
For every $\frac{6}{5}<\de<2$:
\[\al_*\Big(\Big(\frac{3}{2}-\frac{\de}{4},\de\Big),2\Big)\leq 1+\frac{\de}{2}\]
\end{theorem}
\begin{proof}
Denote $\de_1=\frac{3}{2}-\frac{\de}{4}$. {Similarly to}
the proof of Theorem \ref{two_rounds_big_delta_full_thm}, we will use two upper bounds. The first bound is by Lemma \ref{gen_lemma2}, and it will be useful when $m_1\geq \frac{1-\frac{\de}{2}}{1+\frac{\de}{2}}$. The second bound will be a variation of the bound proved in Theorem \ref{two_rounds_small_delta_thm} and will be useful when $m_1\leq \frac{1-\frac{\de}{2}}{1+\frac{\de}{2}}$. Note that this threshold is not coincidental, as it is the value of $m_1$ in Algorithm 3. {By Lemma \ref{gen_lemma2} (using the notation of Theorem \ref{two_rounds_small_delta_thm}) we have that: 
\begin{align*}
    \al_*((\de_1,\de),2)&\leq \max_{p\in\mathcal{S}}\min\Big\{2-(4- 2{\de_1})m_1, \frac{2-(4-2\de)(m_1+m_2)}{e_2}\Big\}=\\
    &=\max_{p\in\mathcal{S}}\min\Big\{2-(4- 2{\de_1})m_1, \frac{\de}{e_2}\Big\}
\end{align*}
\[\] 
}

Denote $\al_1^{\de}(m_1)\coloneqq 2-(4-2\de_1)m_1$. For the second bound, we will use a variation of the proof of Theorem \ref{two_rounds_small_delta_thm}. The only required adjustment to that proof is that the number of $1$-candidate sets for $V_{M_1}$ is at most ${{n^{\de_1}}\choose{m_1k}}\leq n^{\de_1m_1k}$. The rest of the proof does not depend on the number of queries in round $1$, and we get the following bound for this case:
\begin{align*}
    \al_*((\de_1,\de),2)&\leq\max_{m_1}\frac{2(\de_1m_1+\de m_2)}{1-4m_1m_2}=
    \max_{m_1}\frac{2((\frac{3}{2}-\frac{\de}{4})m_1+\de (\frac{1}{2}-m_1))}{1-4m_1(\frac{1}{2}-m_1)}=\\
    &=\max_{m_1}\frac{\de-(\frac{5}{2}\de-3)m_1}{4m_1^2-2m_1+1}
\end{align*}

Denote $\al_2^{\de}(m_1)=\frac{\de-(\frac{5}{2}\de-3)m_1}{4m_1^2-2m_1+1}$. As we did in the proof of Theorem \ref{two_rounds_big_delta_full_thm}, we combine both upper bounds using Lemma \ref{cons_lemma}, {while omitting the expression $\frac{\de}{e_2}$}. We also partition the domain of $m_1$ by the threshold $m_1^*\coloneqq \frac{1-\frac{\de}{2}}{1+\frac{\de}{2}}$, and use a different bound for each case:
\[\al_*((\de_1,\de),2)\leq\max_{m_1}\min\{\al_1^{\de}(m_1),\al_2^{\de}(m_1)\}\leq \max\{\max_{m_1\geq m_1^*}\al_1^{\de}(m_1),\max_{m_1\leq m_1^*}\al_2^{\de}(m_1)\}\]

We start with $\max_{m_1\geq m_1^*}\al_1^{\de}(m_1)$. The expression $\al_1^{\de}(m_1)$ is maximized when $m_1$ is minimal, i.e., for $m_1=m_1^*$. We get:
\begin{align*}
    \max_{m_1\geq m_1^*}\al_1^{\de}(m_1)&= 2-(4- 2{\de_1})\frac{1-\frac{\de}{2}}{1+\frac{\de}{2}}=
    2-\Big(4- \Big(3-\frac{\de}{2}\Big)\Big)\frac{1-\frac{\de}{2}}{1+\frac{\de}{2}}=\\
    &= 2-\Big(1+\frac{\de}{2}\Big)\frac{1-\frac{\de}{2}}{1+\frac{\de}{2}}=
    1+\frac{\de}{2}
\end{align*}

Regarding $\max_{m_1\leq m_1^*}\al_2^{\de}(m_1)$, we prove in Appendix \ref{appendix_two_rounds_calcs} that for $\frac{6}{5}<\de<2$ and $m_1\leq m_1^*$ the maximum of $\al_2^{\de}(m_1)$ is obtained when $m_1= m_1^*$, and in that case $\al_2^{\de}(m_1)=1+\frac{\de}{2}$. Therefore $\max_{m_1\leq m_1^*}\al_2^{\de}(m_1)=1+\frac{\de}{2}$.

We get:
\[\al_*\Big(\Big(\frac{3}{2}-\frac{\de}{4},\de\Big),2\Big)\leq 1+\frac{\de}{2}\]
\end{proof}

Since Algorithm 3 is a $((\de_1,\de),2,1+\frac{\de}{2})$-algorithm for $\de_1=\frac{3}{2}-\frac{\de}{4}$, we have:
\begin{corollary}
For every  $\frac{6}{5}< \de<2$:

\[\al_*\Big(\Big(\frac{3}{2}-\frac{\de}{4},\de\Big),2\Big)= 1+\frac{\de}{2}\]

\end{corollary}

\subsection{Three rounds}
\label{sec_three_rounds}
\subsubsection{Free edges and vertices}
The concept of free edges, introduced in the proof of Theorem \ref{two_rounds_small_delta_thm}, may be extended to any number of rounds $l$. Given a prefect matching $M$ of the returned clique $K$, we partition its edges according to the round in which they were queried, i.e., $M=\bigcup_{i\leq l}M_i$. We denote $|M_i|=m_ik$, and include the parameters $m_1,\dots,m_l$ in the template. We denote by $V_{M_i}$ the set of vertices covered by $M_i$, and denote $V_M^i\coloneqq\bigcup_{j\leq i}V_{M_j}$. Intuitively, an edge $e$ is free for the algorithm, if the algorithm chooses queries such that many candidates for $K$ will contain $e$, based on previous information that $e$ is positive. Since $V(K)=\bigcup_{i\leq l}V_{M_i}$, and the candidates for the set $V_{M_i}$ are chosen in round $i$, if $e$ has some endpoint $u\in V_{M_i}$ and $e$ was answered positively prior to round $i$, then the algorithm may choose to include $u$ in $V_{M_i}$ based on the information that $e$ is positive. Therefore, the algorithm will not 'pay' for the probability $e$ is positive. This intuition suggests that the free edges are exactly the edges that were queried on some round $j$, and have an endpoint $u\in V_{M_i}$ where $i>j$.

In order to prove strong upper bounds on $\mathbb{P}(\mathcal{A}_p \text{ succeeds})$ we want the probability factor of the bound to be small (i.e., $2^{-f_2(p)\kct}$, see Section \ref{constrained algs}). The free edges will not be included in the $f_2(p)\kct$ edges which account for this probability, so we would like to reduce the number of free edges. One way of doing this is by a careful choice of the matching $M$ (as will be done in the proof of Theorem \ref{thm_three_rounds}). {Yet another way is as follows. 
Consider the matching $M^{i^*}\subseteq M$ which contains only matching edges queried up to round $i^*$ for some $1\leq i^*\leq l$, i.e., $M^{i^*}=\bigcup_{j\leq i^*}M_j$. This defines the set $V_M^{i^*}$ of vertices covered by $M^{i^*}$. In order to reduce the number of free edges, we ignore the template parameters $m_{i^*+1},\dots,m_l$, and allow the rest of the vertices to be any set of vertices, regardless of the rounds in which they participate in queries (we used a similar approach in the proof of Lemma \ref{gen_lemma}). This means that the family of candidates for this set of vertices is chosen deterministicly before any answers are given, thus reducing the number of free edges. We denote this set of vertices by $V_{free}^{i^*}$, and refer to them as 'free' vertices for the algorithm, 
{because they are not restricted to participate in any matching. Hence, the family of their candidates will be larger (we will choose single vertices out of $n$ vertices, and not pairs of vertices out of $n^\de$ edges).} Therefore, as we decrease the value of $i^*$, the number of free edges decreases and the number of free vertices increases, which is a trade-off that we may use in order to prove different upper bounds on $\mathbb{P}(\mathcal{A}_p \text{ succeeds})$.}

Determine some $1\leq i^*\leq l$. Given the partition $V(K)=V_M^{i^*}\cup V_{free}^{i^*}= \bigcup_{i\leq i^*}V_{M_i} \cup V_{free}^{i^*}$, for every $1\leq i \leq i^*$ we denote by $e_i\kct$ the number of edges queried on round $i$ that have both endpoints in $V_M^i\cup V_{free}^{i^*}$. The parameters $e_i$, together with the parameters $m_i$, will be a part of the template. We denote the number of free vertices by $|V_{free}^{i^*}|= v_{free}^{i^*}k$. We have that: 
\[v_{free}^{i^*}=2\sum_{i^*<i\leq l}m_i\] Similar to the proof of Theorem \ref{two_rounds_small_delta_thm}, we will define candidate and feasible sets for each round number $i\leq i^*$, and prove an upper bound on the expected number of feasible sets. The definition will be recursive since candidate sets contain feasible sets and vice versa. 
We start with the definition of $0$-candidate sets which are the candidates for the set of free vertices. A $0$-candidate set $a$ satisfies the following: 
\begin{itemize}
    \item $|a|=v_{free}^{i^*}k$.
\end{itemize}

Next, we define $1$-candidate sets (which is the base case of the recursive definition). A $1$-candidate is a set $a$ which may be partitioned into $a=b\sqcup c$ such that:
\begin{itemize}
    \item $b$ is a $0$-candidate.
    \item $b\cap c=\emptyset$.
    \item $|c|=2m_1k$.
    \item $Q_1$ induced on $c$ has a perfect matching.
    \item $Q_1$ induced on $a$ has $e_1\kct$ edges.
\end{itemize}
$a$ will be called $1$-feasible if:
\begin{itemize}
    \item $a$ is a $1$-candidate.
    \item All $e_1\kct$ queries of $Q_1$ induced on $a$ were positive.
\end{itemize}
For $2\leq i \leq i^*$, a set $a$ will be an $i$-candidate if it may be partitioned into disjoint sets $a=b\sqcup c$ which satisfy the following:
\begin{itemize}
    \item $b$ is $i-1$-feasible.
    \item $b\cap c=\emptyset$.
    \item $|c|=2m_ik$.
    \item $Q_i$ induced on $c$ has a perfect matching.
    \item $Q_i$ induced on $a$ has $e_i\kct$ edges.
\end{itemize}
Note that in a general matching $M$, which is not required to be a maximum matching at any stage, we do not require $c$ to be edge-free prior to round $i$ (as opposed to the definition of a candidate pair in the proof of Theorem \ref{two_rounds_small_delta_thm}). For $2\leq i \leq i^*$, a set $a$ will be $i$-feasible if:
\begin{itemize}
    \item $a$ is an $i$-candidate.
    \item All $e_i\kct$ queries of $Q_i$ induced on $a$ were positive.
\end{itemize}

Denote by $F^i$ the number of $i$-feasible sets. Then we have the following lemma:
\begin{lemma}
\label{free_edges_vers_lemma}
For every $1\leq i^*\leq l$ and $1\leq i\leq i^*$, assuming $q_i=n^{\de_i}$, we have:
\[\E_{1,\dots,i}[F^i]\leq n^{\Big[v_{free}^{i^*}+\sum_{j\leq i}\de_jm_j\Big]k}\cdot 2^{-\sum_{j\leq i}e_j\kct}\]
The subscript indicates the rounds over which the expectation is taken.
\end{lemma}
\begin{proof}
The proof is by induction on the round number $i$. For $i=1$, the number of $0$-candidates is at most ${{n}\choose{v_{free}^{i^*}k}}\leq n^{v_{free}^{i^*}k}$, and the number of $1$-candidates is at most $n^{v_{free}^{i^*}k}\cdot {{n^\de}\choose{m_1k}} \leq n^{(v_{free}^{i^*}+\de_1m_1)k} $. The family of $1$-candidates is determined deterministically by the algorithm before it receives answers. Hence, each candidate is $1$-feasible with probability $2^{-e_1\kct}$, so we get the desired bound:
\[\E_{1}[F^1]\leq n^{\Big[v_{free}^{i^*}+\de_1m_1\Big]k}\cdot 2^{-e_1\kct}\]

Assuming the claim is correct for some $i-1$, we wish to prove it for $i$. An $i$-candidate $a$ is the union of two sets $b\cup c$ where $b$ is an $i-1$-feasible set, and $c$ is a set with a perfect matching in $Q_i$. After round $i$ there are at most $F^{i-1}\cdot n^{\de_im_ik}$ such pairs, and the probability such a pair induces an $i$-feasible set is at most $2^{-e_i\kct}$ independently of the answers up to round $i-1$ (since this probability comes from round $i$ answers). By linearity of expectation we get:
\begin{align*}
    \E_{1,\dots,i}[F^i]&=\E_{1,\dots,i-1}[\E_{i}[F^i]]\leq \E_{1,\dots,i-1}[F^{i-1}\cdot n^{\de_im_ik}\cdot 2^{-e_i\kct}]=\\
    &=n^{\de_im_ik}\cdot 2^{-e_i\kct}\cdot\E_{1,\dots,i-1}[F^{i-1}]\leq n^{\Big[v_{free}^{i^*}+\sum_{j\leq i}\de_jm_j\Big]k}\cdot 2^{-\sum_{j\leq i}e_j\kct}
\end{align*}
\end{proof}
For $i=i^*$ we get an upper bound on $\E_{1,\dots,i^*}[F^{i^*}]$, and since the clique $K$ returned by the algorithm must be an $i^*$-feasible set, we deduce a bound on the probability such a clique exists. 

As we defined in the intuitive explanation, an edge $e$ queried on round $i$ for which $i<i^*$, will be a free edge if it has an endpoint $u$ in $V_{M_j}$ where $i<j\leq i^*$. This happens if and only if the edge is not contained in $V_M^i$, i.e., it is not one of the $e_i\kct$ edges accounted for the probability factor of the bound in Lemma \ref{free_edges_vers_lemma}. Therefore, denoting by $e_{free}^{i^*}\kct$ the number of free edges, we have that:
\[e_{free}^{i^*}=1-\sum_{1\leq i\leq i^*}e_i\]
Using Lemma \ref{cons_lemma} with the template $m_1,\dots,m_{i^*}, e_1,\dots,e_{i^*}$, and the $l$ bounds proved in Lemma \ref{free_edges_vers_lemma} (i.e., for any $1\leq i^*\leq l$) we get the following corollary:

\begin{corollary}

\label{free_edges_vertices_cor}
\[\al_*((\de_1,\dots,\de_l),l)\leq \max_{p\in \mathcal{S}}\min_{1\leq i^*\leq l}\frac{2\Big[v_{free}^{i^*}+\sum_{i\leq i^*}\de_im_i\Big]}{1-e_{free}^{i^*}}\]

\end{corollary}

\subsubsection{Three rounds analysis}
\label{sec_three_rounds_analysis}

In the following sections we analyze the case of three rounds algorithms for any $1\leq \de < 2$, and prove upper bounds on $\al_*(\de,3)$ using Corollary \ref{free_edges_vertices_cor}. We introduce a specific perfect matching $M$ of $K$, and prove upper bounds on the number of free edges corresponding to it. $M$ partitions $V(K)$ into three sets $V_{M_1},V_{M_2},V_{M_3}$, where $|V_{M_i}|=2m_ik$, and we have $m_1+m_2+m_3=\frac{1}{2}$. We will use the following notation: $M_{ij}=M_i\cup M_j$ and $V_{M_{ij}}=V_{M_i}\cup V_{M_j}$. $M$ is defined as the perfect matching for which $m_1+m_2$ is maximal, breaking ties by maximizing the value of $m_1$ (and then breaking ties arbitrarily). We consider $m_1,m_2,m_3$ of the matching $M$ defined this way, as parameters in the template. The reason for choosing the matching $M$ this way is that we get a good upper bound on the total number of free edges. 
For example, we have the following upper bound:

\begin{lemma}
\label{free_edges_at_most_half}
\[e_{free}^3\leq \frac{1}{2}\]
\end{lemma}
\begin{proof}

Let $\{u,v\},\{w,z\}$ be two edges in the matching $M$. Between the vertices $u,v,w,z$ there are four non-matching edges (i.e., not in $M$) - $\{u,w\}$, $\{u,z\}$, $\{v,w\}$, $\{v,z\}$, which are adjacent to both matching edges. For any non-matching edge $e$, there is exactly one pair of distinct matching edges adjacent to $e$. Therefore, going over all pairs of matching edges, and the four edges adjacent to them, we count each non-matching edge exactly once. Since the fraction of matching edges is negligible, it is enough to prove that for any pair of matching edges, at most two edges out of $\{u,w\}$, $\{u,z\}$, $\{v,w\}$, $\{v,z\}$ are free. 

We prove that at most one of the edges $\{u,w\},\{v,z\}$ is free (the proof for the other pair is equivalent). Assume for the sake of contradiction that both edges are free. {We will introduce a matching $\Bar{M}$ that contradicts the maximality properties of $M$.} There are no free edges queried in round $3$, so the edges are in $Q_1\cup Q_2$. We divide into the following cases:

\begin{itemize}
    \item \textbf{Both $\{u,w\},\{v,z\}$ were queried in round 1.} {Suppose w.l.o.g that $\{u,w\}$ is free because $u$ is covered by $M$ after $\{u,w\}$ is answered, i.e., $\{u,v\}\in M_{23}$, so $\{u,v\}\in Q_2\cup Q_3$.}
    We define a matching $\Bar{M}$ which is equal to $M$ except on the vertices $u,v,w,z$. Instead of the edges $\{u,v\},\{w,z\}$, $\Bar{M}$ contains $\{u,w\},\{v,z\}$. Denote by $\Bar{m}_1,\Bar{m}_2,\Bar{m}_3$ the relative parts of the partition of $\Bar{M}$ to rounds. Since $\{u,w\},\{v,z\}\in Q_1$ and $\{u,v\}\in Q_2\cup Q_3$ we have:
    \begin{align*}
        \Bar{m}_1+\Bar{m}_2&\geq m_1+m_2\\
        \Bar{m}_1&>m_1
    \end{align*}
    contradicting the maximality of $M$ (which maximizes $m_1+m_2$, breaking ties by maximizing $m_1$).
    
    \item \textbf{At least one edge was queried on round 2}. Suppose w.l.o.g that $\{u,w\}\in Q_2$ and it is free because $u$ is covered by $M$ after $\{u,w\}$ is answered, i.e., $\{u,v\}\in M_{3}$, so $\{u,v\}\in Q_3$. We use the same matching $\Bar{M}$ as in the previous case. Since $\{u,w\},\{v,z\}\in Q_1\cup Q_2$, and $\{u,v\}\in Q_3$ we have:
    \[\Bar{m}_1+\Bar{m}_2> m_1+m_2\]
    contradicting the maximality of $M$.
\end{itemize}
Hence, at most one edge of each pair is free and $e_{free}^3\leq \frac{1}{2}$.
\end{proof}

The edges in $M_{12}$ form a maximum matching in $K_2$ (since $m_1+m_2$ is maximal), so we can look at the Gallai-Edmonds decomposition after round $2$. It is composed of the sets $C,S,R$, or the specific matching representation - $C^{*+},C^{*-},S,D$ where $C^{*}=C^{*+}\cup C^{*+}$ (as defined in Section \ref{sec_ged}). 
We partition further according to the partition of $V(K)$ into $V_{M_1}\cup V_{M_2}\cup V_{M_3}$. We get the following sets (see Figure \ref{partition1}):
\begin{itemize}
    \item $S$ decomposes to $S_1\cup S_2$ where $S_i=S\cap V_{M_i}$.
    \item $C^{*+}$ decomposes to the sets $C^{*+}_1, C^{*+}_2$ which are matched to $S_1, S_2$ respectively.
    \item $C^{*-}=V_{M_3}$. In particular, $C^{*-}$ contains all isolated vertices in $K_2$.
    \item The set $D$ decomposes to $D_1\cup D_2$ where $D_i=D\cap V_{M_i}$. 
\end{itemize}

\begin{figure}[h!]
  \centering
  \includegraphics[width=1\textwidth]{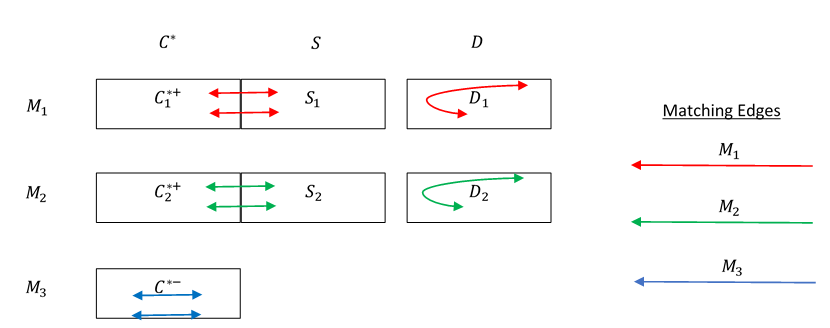}
  \caption{Specific matching decomposition of $K$. The arrows indicate the partition of $M$ according to the rounds and the decomposition.}
  \label{partition1}
\end{figure}

The edges in $M_2$ form a maximum matching in $V_{M_{23}}$ after round $2$, because if there was a larger matching we could add $M_1$ and get a matching larger than $M_{12}$ in $K_2$. Therefore we can look at the GED of $V_{M_{23}}$, and the specific matching decomposition. We denote the relevant sets by $\Tilde{C},\Tilde{S},\Tilde{R}$. We have that $\Tilde{C}\subseteq C\cap V_{M_{23}}$ because for any $v\in \Tilde{C}$ there is a maximum matching in $V_{M_{23}}$ in which $v$ does not participate. If we add $M_1$ to this matching, we get a maximum matching in $K$ in which $v$ does not participate, so $v \in C$. Therefore we have that $\Tilde{S}=N(\Tilde{C})\subseteq N(C)=S$, so $\Tilde{S}\subseteq S\cap V_{M_{23}}=S_2$, and also $R\cap V_{M_{23}}\subseteq \Tilde{R}$. Looking at the specific matching decomposition $\Tilde{C}^{*}=\Tilde{C}^{*+}\cup \Tilde{C}^{*-},\Tilde{S},\Tilde{D}$, we have $\Tilde{C}^{*-}=V_{M_3}$ and $\Tilde{S}$ is matched to $\Tilde{C}^{*+}$. Since $\Tilde{S}\subseteq S_2$ and $\Tilde{C}^*\subseteq C^*\cap V_{M_{23}}$ we get $D_2\subseteq \Tilde{D}$, so we can write $\Tilde{D}$ as $D_2\cup \Tilde{D_2}$ (see Figure \ref{partition2}).

For each vertex $v\in V(K)$, and $K'\subseteq K$ we denote the degree of $v$ in the graph $Q_i$ induced on $V(K')$ by $deg_i^{K'}(v)$. 
\begin{lemma}
\label{deg1_lemma}
For every $v\in \Tilde{C}$, $deg_1^{V_{M_{23}}}(v)=0$.
\end{lemma}
\begin{proof}
First, suppose for contradiction that for some $v\in V_{M_3}$, $deg_1^{V_{M_{23}}}(v)\geq 1$. All edges within $\Tilde{C}^*$ are in $Q_3$, so there must be $u\in \Tilde{S}\cup \Tilde{D}$ s.t. $\{u,v\}\in Q_1$. Let $u',v'$ be the vertices matched to $u,v$. We have $\{u,u'\}\in Q_2$ and $\{v,v'\}\in Q_3$. Hence, by replacing the matching edges with $\{u,v\}$ and $\{u',v'\}$ in the original full matching of $K$, we get a matching in which $m_1+m_2$ is not smaller, but $m_1$ is larger, which is a contradiction. 

Now suppose there is $v\in \Tilde{C}$ such that  $deg_1^{V_{M_{23}}}(v)\geq 1$. This means there is $u\in V_{M_{23}}$ s.t. $\{u,v\}\in Q_1$. $v\in \Tilde{C}$ so by definition there is maximum matching of $V_{M_{23}}$ in which $v$ is not matched. In this new matching, which might consist of a different decomposition -  $\Tilde{C}^*_{new},\Tilde{S},\Tilde{D}_{new}$, it must be that $u\notin \Tilde{C}^*_{new}$, because there are only $Q_3$ edges in $\Tilde{C}^*_{new}$. So $u\in \Tilde{S} \cup \Tilde{D}_{new}$. Now we can apply the previous argument about $v\in V_{M_3}$ but using the new matching and get a contradiction to $\{u,v\}\in Q_1$.
\end{proof}

Another parameter of the decomposition that we will use in the template, is the size of the maximum $Q_1$ matching in $V_{M_{23}}$, denoted $M_1'$ (see Figure \ref{partition2}). The set of vertices covered by this matching will be denoted $V_{M_1'}$, and its size will be $|V_{M_1'}|=2m_1'k$. Note that by Lemma \ref{deg1_lemma}, $V_{M_1'}\subseteq\Tilde{S}\cup\Tilde{R}$.

\begin{figure}[h!]
  \centering
  \includegraphics[width=1\textwidth]{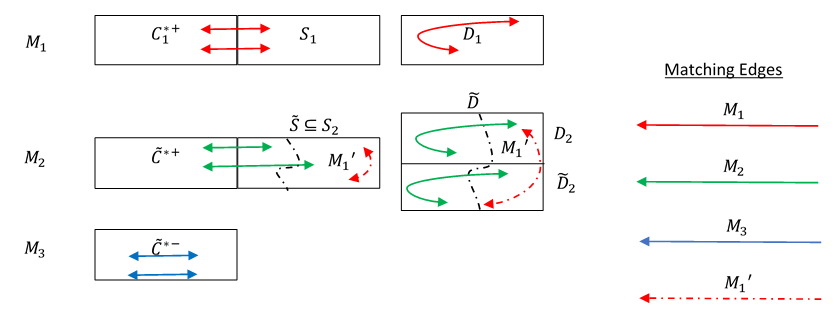}
  \caption{Specific matching decomposition of $V_{M_{23}}$ and $K$, including $M_1'$. }
  \label{partition2}
\end{figure}

We will use the notation of $E_i(A,B)$ and $E_i(A)$ as the number of $Q_i$ edges between $A$ and $B$ or within $A$. If no index is mentioned it means all edges between the sets.

The decomposition described is rich enough to extract strong bounds on the number of free edges. In order to do so, we replace the template parameters $m_1,m_2$ by the parameters $s_1,d_1,\st,\dt$ where $|S_1|=s_1k$, $|D_1|=2d_1k$, $|\Tilde{S}|=\st k$, $|\Tilde{D}|=2\dt k$. Note that $m_1=s_1+d_1$ and $m_2=\st+\dt$. In the three rounds setting, free edges may be in either $Q_1$ or $Q_2$, and we start with the analysis of round 2 free edges, which is simpler (see Figure \ref{free_egdes_fig} for a summary of all free edges in $K$). 

\paragraph*{Round 2 free edges}
Any free edge of $Q_2$ must have an endpoint in $V_{M_3}$, which equals $C^{*-}$ in the decomposition. Looking at $K_2$, $C^{*-}$ is an independent set and has all neighbors in $S\cup D$. Moreover, any $v\in D$ satisfies $\deg_2^{C^{*-}}(v)\leq 1$ (because the vertices of $C^{*-}$ belong to different connected components in the G.E.D). Hence, $E_2(C^{*-},D)$ is at most linear in $k$, and so neglecting low order terms we may assume the number of free edges of $Q_2$ is bounded by $E_2(C^{*-},S)$. This can be further reduced by looking at $M_2$ as a maximal matching of $V_{M_{23}}$ after round 2. All $Q_2$ edges in $V_{M_{23}}$ with an endpoint in $C^{*-}=\Tilde{C}^{*-}$ have the other endpoint in $\Tilde{S}\subseteq S_2$ (up to a linear in $k$ amount). Hence, the number of such free edges is bounded by \[E_2(C^{*-},S_1\cup \Tilde{S})\leq 2m_3k\cdot (s_1+\st)k\simeq4m_3(s_1+\st)\kct\]

\paragraph*{Round 1 free edges}
Free edges of $Q_1$ must have an endpoint in either $V_{M_3}$ or $V_{M_2}$. A free edge $e$ that has an endpoint in $V_{M_3}$, must have the other endpoint in $V_{M_1}$ (by Lemma \ref{deg1_lemma}). So up to low order terms, the number of such edges is at most 
\[E_1(C^{*-},S_1)\leq 2m_3k\cdot s_1k\simeq4m_3s_1\kct\]

Note that when we bound both $Q_1$ and $Q_2$ free edges, the $Q_1$ bound is redundant because we have: 
\[E_2(C^{*-},S_1\cup \Tilde{S})+E_1(C^{*-},S_1)\leq E(C^{*-},S_1\cup \Tilde{S}) \leq 4m_3(s_1+\st)\kct\]

Assuming a free edge $e$ has no endpoints in $V_{M_3}$, it must have an endpoint in $V_{M_2}$. For the case where the other endpoint is in $V_{M_1}$ we have to bound $E_1(V_{M_1}, V_{M_2})$. For any pair of matching edges $(e_1,e_2)\in M_1\times M_2$, there are four non-matching edges between the four vertices incident to $e_1,e_2$, and at most two of those edges are in $Q_1$.
Otherwise, we could enlarge $m_1$ by using two edges of $Q_1$ instead of $(e_1,e_2)$. Going over all pairs $(e_1,e_2)\in M_1\times M_2$ and the edges between them, we count all edges in $E(V_{M_1},V_{M_2})$ exactly once, so at most half of them are in $Q_1$. Hence:
\[E_{1}(V_{M_1},V_{M_2}) \leq \frac{1}{2}\cdot 2m_1k\cdot 2m_2k\simeq 4(s_1+d_1)(\st+\dt)\kct\]

For the last case, where both endpoints are in $V_{M_2}$, we have to bound $E_1(V_{M_2})$. By Lemma \ref{deg1_lemma}, $Q_1$ edges inside $V_{M_2}$ can be only in $\Tilde{S}\cup \Tilde{D}$. We partition $\Tilde{S}\cup \Tilde{D}$ into $V_{M_1'}$ and $[\Tilde{S}\cup \Tilde{D}]\setminus{V_{M_1'}}$ ( $|[\Tilde{S}\cup \Tilde{D}]\setminus{V_{M_1'}}|=(\st+2\dt-2m_1')k$). Since $M_1'$ is the maximum $Q_1$ matching in $V_{M_2}$, there are no $Q_1$ edges within $[\Tilde{S}\cup \Tilde{D}]\setminus{V_{M_1'}}$, and at most half the edges between the two sets are from $Q_1$. This is because if we look at the GED of $Q_1$ in $V_{M_2}$, at most one endpoint of each maximum matching edge can be connected to an unmatched vertex. Therefore there are at most $\frac{1}{2}\cdot 2m_1'k\cdot (\st+2\dt-2m_1')k\simeq 2m_1'(\st+2\dt-2m_1')\kct$ such edges. The rest of the edges are those within $V_{M_1'}$. There are ${{2m_1'k}\choose{2}}\simeq 4{m'_1}^2\kct$ such edges, all may be in $Q_1$. We get 
\[E_{1}(V_{M_2})\leq [4{m'_1}^2+2m_1'(\st+2\dt-2m_1')]\kct=2m_1'(\st+2\dt)\kct\]

\begin{figure}[h!]
  \centering
  \includegraphics[width=1\textwidth]{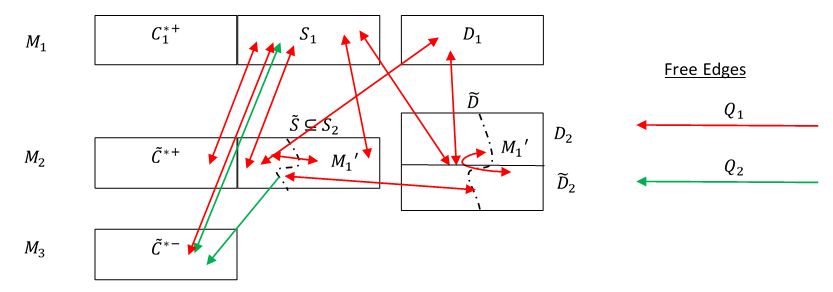}
  \caption{Free edges in $K$}
  \label{free_egdes_fig}
\end{figure}

\begin{corollary}
\label{free_edges_bounds_cor}
We have the following bounds on $e_{free}^2,e_{free}^3$:
\[e_{free}^2\leq 2[2(s_1+d_1)(\st+\dt)+m_1'(\st+2\dt)]\]
\[e_{free}^3\leq 2[2m_3(s_1+\st)+2(s_1+d_1)(\st+\dt)+m_1'(\st+2\dt)]\]
\end{corollary}

\subsubsection{Tight bound for three rounds, restricted first rounds}
Algorithm 4 (see Section \ref{algos_sec}) provides the lower bound $\al_*(\de,3)\geq 1+\frac{\de}{2}$ for every $1\leq\de<2$. As we did in Section \ref{sec_two_rounds_large_delta_restricted}, we will show that Algorithm 4 is optimal among all algorithms which query the same number of queries per round as Algorithm 4 does. Algorithm 4 queries roughly $n^{1-\frac{\de}{2}}$ pairs in round 1, $\theta (n)$ pairs in round 2, and $n^\de$ pairs in round 3. We therefore have $(\de_1,\de_2,\de_3)=(1-\frac{\de}{2},1,\de)$. We prove the following theorem:
\begin{theorem}
\label{thm_three_rounds}
For every $1\leq\de<2$:
\[\al_*((1-\frac{\de}{2},1,\de),3)\leq 1+\frac{\de}{2}\]
\end{theorem}

\begin{proof}

By Corollary \ref{free_edges_vertices_cor} and Lemma \ref{max_indep_lemma}, we get the following upper bound:
\[\al_*((1-\frac{\de}{2},1,\de),3)\leq \max_{p\in \mathcal{S}_{I_{l-1}\leq \de\log n}}\min\{\al_1,\al_2,\al_3\}\]
where $\al_{i^*}$ is the upper bound corresponding to round $i^*$. We use the upper bounds proved in Corollary \ref{free_edges_bounds_cor} on the number of free edges in order to upper bound $\al_2,\al_3$.

\begin{itemize}
    \item For $i^*=3$, i.e., there are no free vertices, we get:
    \begin{align*}
        \al_3&= \frac{2\big[(1-\frac{\de}{2})m_1+m_2+\de m_3\big]}{1-e_{free}^3}=\\
        &=\frac{2\big[(1-\frac{\de}{2})m_1+m_2+\de (\frac{1}{2}-m_1-m_2)\big]}{1-e_{free}^3}=\\
        &=\frac{\de-(3\de-2)m_1-(2\de-2)m_2}{1-e_{free}^3}\leq \\
        &\leq \frac{\de-(3\de-2)(s_1+d_1)-(2\de-2)(\st+\dt)}{1-2[2m_3(s_1+\st)+2(s_1+d_1)(\st+\dt)+m_1'(\st+2\dt)]}
    \end{align*}
        
    \item For $i^*=2$, i.e., the free vertices are those matched in round 3, we get:
    \begin{align*}
        \al_2&= \frac{2\big[2m_3+(1-\frac{\de}{2})m_1+m_2\big]}{1-e_{free}^2}=\\
        &= \frac{2\big[2(\frac{1}{2}-m_1-m_2)+(1-\frac{\de}{2})m_1+m_2\big]}{1-e_{free}^2}\leq\\
        &\leq\frac{2-(2+\de)(s_1+d_1)-2(\st+\dt)}{1-2[2(s_1+d_1)(\st+\dt)+m_1'(\st+2\dt)]}
    \end{align*}
    \item For $i^*=1$, i.e., the free vertices are those matched in rounds 2,3, and there are no free edges. We get:
    \begin{align*}
        \al_1&= 2(1-(2-\de_1)m_1) =2-(2+\de)m_1
    \end{align*}
    This bound is stronger when $m_1$ is larger. $M_1$ is not a maximal matching in $K_1$, so using the matching $M$ in this case is not optimal. Instead, we look at $M_1\cup M_1'$ which is a maximal matching in $K_1$, of size $(m_1+m_1')k$. We get the following stronger bound:
    \begin{align*}
        \al_1&= 2-(2+\de)(m_1+m_1')=2-(2+\de)(s_1+d_1+m_1')
    \end{align*}
    
\end{itemize}

By Lemma \ref{max_indep_lemma}, the signatures set is restricted to the signatures in which $I_{l-1}\leq \de\log n$. Since the set $C_1^{*+}\cup \Tilde{C}^{*+}\cup \Tilde{C}^{*-}$ is an independent set in $K_2$, we have $(s_1+\st+2m_3)k\leq I_{l-1}\leq \de\log n$ and therefore: $\al\leq\frac{\de}{s_1+\st+2m_3}$. Using the lower bound $\al_*((1-\frac{\de}{2},1,\de),3)\geq 1+\frac{\de}{2}$ implied by Algorithm 4, we may assume that $\al\geq 1+\frac{\de}{2}$. We get:
\begin{gather*}
    1+\frac{\de}{2}\leq \frac{\de}{s_1+\st+2m_3} \\
    s_1+\st+2(\frac{1}{2}-(s_1+d_1+\st+\dt))\leq \frac{2\de}{2+\de}
\end{gather*}
\begin{equation}
    \Rightarrow s_1+2d_1+\st+2\dt\geq \frac{2-\de}{2+\de}
    \label{eq:max_indep_constraint} \tag{$\ast$}
\end{equation}

We now wish to find the signature which satisfies the constraint \eqref{eq:max_indep_constraint} and maximizes $\min\{\al_1,\al_2,\al_3\}$. We start with the value of $d_1$, and prove the following lemma:

\begin{lemma}
\label{lemma_d_1_equals_0}
The maximum of $\min\{\al_1,\al_2,\al_3\}$ under the following relaxation of \eqref{eq:max_indep_constraint}: 
\begin{equation}
    2s_1+2d_1+\st+2\dt\geq \frac{2-\de}{2+\de}
    \label{eq:relaxed_max_indep_constraint} \tag{$\ast\ast$}
\end{equation}
is obtained when $d_1=0$.
\end{lemma}
\begin{proof}
Suppose the maximum of $\min\{\al_1,\al_2,\al_3\}$ is obtained for some signature $(s_1,d_1,\st,\dt,m_1')=(a,b,c,d,e)$ where $b>0$. We claim that for $(s_1,d_1,\st,\dt,m_1')=(a+b,0,c,d,e)$, $\min\{\al_1,\al_2,\al_3\}$ is not smaller and \eqref{eq:relaxed_max_indep_constraint} is satisfied. Note that $m_1=s_1+d_1$ is the same in both signatures. Moreover, $\al_1,\al_2$ and the constraint \eqref{eq:relaxed_max_indep_constraint} may be expressed with $m_1$ instead of $s_1$ and $d_1$ (so they are the same for both signatures). Regarding $\al_3$, it is larger when $(s_1,d_1)=(a+b,0)$ than when $(s_1,d_1)=(a,b)$. Hence, $\min\{\al_1,\al_2,\al_3\}$ is not smaller using the signature $(a+b,0,c,d,e)$, and the constraint \eqref{eq:relaxed_max_indep_constraint} is satisfied. (Note that this may not hold for the original constraint \eqref{eq:max_indep_constraint}.)
\end{proof}

According to Lemma \ref{lemma_d_1_equals_0} we assume $d_1=0$, and maximize $\min\{\al_1,\al_2,\al_3\}$ under the relaxed constraint \eqref{eq:relaxed_max_indep_constraint}. We can simplify the bounds further by using the lower bound $\al_*((1-\frac{\de}{2},1,\de),3)\geq 1+\frac{\de}{2}$ implied by Algorithm 4. Any signature $p$ for which the expression that bounds $\al_1$ is smaller than $1+\frac{\de}{2}$, is necessarily not the signature which maximizes the objective. Hence, we may assume that:

\begin{align*}
    &1+\frac{\de}{2}\leq 2-(2+\de)(s_1+m_1')\\
    \Rightarrow &m_1'\leq \frac{1-\frac{\de}{2}}{2+\de}-s_1
\end{align*}

Both expressions that bound $\al_2,\al_3$ are larger when $m_1'$ is larger. Therefore the bounds hold also when we assume $m_1'$ is maximal. We also have that $m_3=\frac{1}{2}-(s_1+\st+\dt)$. We get the bound $\al_*((1-\frac{\de}{2},1,\de),3)\leq \max_{s_1,\st,\dt}\min\{\al_2,\al_3\}$ when:

\begin{gather*}
    \al_3\leq \frac{\de-(3\de-2)s_1-(2\de-2)(\st+\dt)}{1-2[2(\frac{1}{2}-(s_1+\st+\dt))(s_1+\st)+2s_1(\st+\dt)+\big[\frac{1-\frac{\de}{2}}{2+\de}-s_1\big](\st+2\dt)]}\\
    \al_2\leq \frac{2-(2+\de)s_1-2(\st+\dt)}{1-2[2s_1(\st+\dt)+\big[\frac{1-\frac{\de}{2}}{2+\de}-s_1\big](\st+2\dt)]} 
\end{gather*}
Under the constraints:
\begin{gather*}
    s_1,\st,\dt\geq 0 \\
    s_1+\st+\dt\leq \frac{1}{2}\\
    2s_1+\st+2\dt\geq \stg
\end{gather*}

We refer the reader to Appendix \ref{appendix_optimization} for a description of the optimization method we used to solve this problem. The signature which maximizes the objective is as follows:  $s_1=0$, $\st=\frac{2-\de}{2+\de}$, $\dt=0$. For this signature, which is equal to the one used in Algorithm 4, the value of $\min(\al_2,\al_3)$ is $1+\frac{\de}{2}$. Hence, for every $1\leq\de<2$:
\[\al_*((1-\frac{\de}{2},1,\de),3)\leq 1+\frac{\de}{2}\]
\end{proof}
Since Algorithm 4 is a $((1-\frac{\de}{2},1,\de),3,1+\frac{\de}{2})$-algorithm, we have:
\begin{corollary}
For every  $1\leq \de<2$:

\[\al_*((1-\frac{\de}{2},1,\de),3)= 1+\frac{\de}{2}\]
\end{corollary}

\subsubsection{Bound for three rounds, $q_1=q_2=q_3$}
\label{sec_three_rounds_non_restricted}

In this section we give upper bounds on $\al_*(\de,3)$ (i.e., when no round is restricted) for several values of $\de$. The bounds are achieved using Corollary \ref{free_edges_vertices_cor}, Lemma \ref{max_indep_lemma}, and Corollary \ref{free_edges_bounds_cor}. 
\[\al_*(\de,3)\leq \max_{p\in \mathcal{S}_{I_{l-1}\leq \de\log n}}\min\{\al_1,\al_2,\al_3\}\]
Where $\al_{i^*}$ is the upper bound corresponding to round $i^*$. 

\begin{itemize}
    \item For $i^*=3$, i.e., there are no free vertices, we get:
    \begin{align*}
        \al_3&= \frac{2\big[\de m_1+ \de m_2+\de m_3\big]}{1-e_{free}^3}\leq \\
        &\leq \frac{\de}{1-2[2m_3(s_1+\st)+2(s_1+d_1)(\st+\dt)+m_1'(\st+2\dt)]}
    \end{align*}
        
    \item For $i^*=2$, i.e., the free vertices are those matched in round 3, we get:
    \begin{align*}
        \al_2&= \frac{2\big[2m_3+\de m_1+\de m_2\big]}{1-e_{free}^2}=\\
        &= \frac{2\big[2(\frac{1}{2}-m_1-m_2)+\de m_1+\de m_2\big]}{1-e_{free}^2}\leq\\
        &\leq\frac{2-(4-2\de)(s_1+d_1+\st+\dt)}{1-2[2(s_1+d_1)(\st+\dt)+m_1'(\st+2\dt)]}
    \end{align*}
    \item For $i^*=1$, as we did in the proof of Theorem \ref{thm_three_rounds}, we use the matching $M_1\cup M_1'$ instead of $M_1$. We get the bound:
    \begin{align*}
        \al_1&= 2(1-(2-\de)(m_1+m_1')) = 2-(4-2\de)(s_1+d_1+m_1')
    \end{align*}
    
\end{itemize}
{Note that the arguments we used in the proof of Theorem \ref{thm_three_rounds} regarding Lemma \ref{max_indep_lemma}, the constraint \eqref{eq:max_indep_constraint} and Lemma \ref{lemma_d_1_equals_0}, hold similarly in the current analysis (i.e., when $q_1=q_2=q_3$), so we may assume $d_1=0$ and maximize $\min\{\al_1,\al_2,\al_3\}$ under the constraint \eqref{eq:relaxed_max_indep_constraint}.} We get the following optimization problem:

$\al_*(\de,3)\leq \max_{s_1,\st,\dt,m_1'}\min\{\al_1,\al_2,\al_3\}$ when:
\begin{gather*}
    \al_3\leq \frac{\de}{1-2[2(\frac{1}{2}-(s_1+\st+\dt))(s_1+\st)+2s_1(\st+\dt)+m_1'(\st+2\dt)]}\\
    \al_2\leq \frac{2-(4-2\de)(s_1+\st+\dt)}{1-2[2s_1(\st+\dt)+m_1'(\st+2\dt)]} \\
    \al_1=2-(4-2\de)(s_1+m_1')
\end{gather*}
Under the constraints:
\begin{gather*}
    s_1,\st,\dt,m_1'\geq 0 \\
    s_1+\st+\dt\leq \frac{1}{2}\\
    m_1'\leq \frac{\st}{2}+\dt \\
    2s_1+\st+2\dt\geq \stg
\end{gather*}

We solved the optimization problem for several values of $\de$, and we refer the reader to Appendix \ref{appendix_optimization_non_restriceted} for more details. We received the following upper bounds:

\begin{center}
\begin{tabular}{ |c||c c c c c c c c c c| } 
 \hline
 $\de$ & 1 & 1.1 & 1.2 & 1.3 & 1.4 & 1.5 & 1.6 & 1.7 & 1.8 & 1.9  \\ 
 $\al_*(\de,3)\leq$ & 1.62 & 1.69 & 1.77 & 1.83 & 1.88 & 1.92 & 1.95 & 1.97 & 1.99 & 1.997\\ 

 \hline
\end{tabular}
\end{center}


\section*{Acknowledgements}

This project has received funding from the European Research Council (ERC) under the European Union’s Horizon 2020 research and innovation programme (grant agreement No. 819702)

\clearpage
\printbibliography

@article{feige2020finding,
  title={Finding cliques using few probes},
  author={Feige, Uriel and Gamarnik, David and Neeman, Joe and R{\'a}cz, Mikl{\'o}s Z and Tetali, Prasad},
  journal={Random Structures \& Algorithms},
  volume={56},
  number={1},
  pages={142--153},
  year={2020},
  publisher={Wiley Online Library}
}

@inproceedings{rivest1975generalization,
  title={A generalization and proof of the Aanderaa-Rosenberg conjecture},
  author={Rivest, Ronald L and Vuillemin, Jean},
  booktitle={Proceedings of the seventh annual ACM symposium on Theory of computing},
  pages={6--11},
  year={1975}
}

@article{ferber2016finding,
  title={Finding Hamilton cycles in random graphs with few queries},
  author={Ferber, Asaf and Krivelevich, Michael and Sudakov, Benny and Vieira, Pedro},
  journal={Random Structures \& Algorithms},
  volume={49},
  number={4},
  pages={635--668},
  year={2016},
  publisher={Wiley Online Library}
}

@article{ferber2017finding,
  title={Finding paths in sparse random graphs requires many queries},
  author={Ferber, Asaf and Krivelevich, Michael and Sudakov, Benny and Vieira, Pedro},
  journal={Random Structures \& Algorithms},
  volume={50},
  number={1},
  pages={71--85},
  year={2017},
  publisher={Wiley Online Library}
}

@incollection{conlon2019online,
  title={Online Ramsey Numbers and the Subgraph Query Problem},
  author={Conlon, David and Fox, Jacob and Grinshpun, Andrey and He, Xiaoyu},
  booktitle={Building Bridges II},
  pages={159--194},
  year={2019},
  publisher={Springer}
}

@article{racz2019finding,
  title={Finding a planted clique by adaptive probing},
  author={R{\'a}cz, Mikl{\'o}s Z and Schiffer, Benjamin},
  journal={arXiv preprint arXiv:1903.12050},
  year={2019}
}

@article{mardia2020finding,
  title={Finding Planted Cliques in Sublinear Time},
  author={Mardia, Jay and Asi, Hilal and Chandrasekher, Kabir Aladin},
  journal={arXiv preprint arXiv:2004.12002},
  year={2020}
}

@article{scheidweiler2013lower,
  title={A lower bound for the complexity of monotone graph properties},
  author={Scheidweiler, Robert and Triesch, Eberhard},
  journal={SIAM Journal on Discrete Mathematics},
  volume={27},
  number={1},
  pages={257--265},
  year={2013},
  publisher={SIAM}
}

@article{kahn1984topological,
  title={A topological approach to evasiveness},
  author={Kahn, Jeff and Saks, Michael and Sturtevant, Dean},
  journal={Combinatorica},
  volume={4},
  number={4},
  pages={297--306},
  year={1984},
  publisher={Springer}
}

@article{yao1988monotone,
  title={Monotone bipartite graph properties are evasive},
  author={Yao, Andrew Chi-Chih},
  journal={SIAM Journal on Computing},
  volume={17},
  number={3},
  pages={517--520},
  year={1988},
  publisher={SIAM}
}

@article{yao1991lower,
  title={Lower bounds to randomized algorithms for graph properties},
  author={Yao, Andrew Chi-Chih},
  journal={Journal of Computer and System Sciences},
  volume={42},
  number={3},
  pages={267--287},
  year={1991},
  publisher={Elsevier}
}

@inproceedings{chakrabarti2001improved,
  title={Improved lower bounds on the randomized complexity of graph properties},
  author={Chakrabarti, Amit and Khot, Subhash},
  booktitle={International Colloquium on Automata, Languages, and Programming},
  pages={285--296},
  year={2001},
  organization={Springer}
}

@inproceedings{friedgut2002computing,
  title={Computing graph properties by randomized subcube partitions},
  author={Friedgut, Ehud and Kahn, Jeff and Wigderson, Avi},
  booktitle={International Workshop on Randomization and Approximation Techniques in Computer Science},
  pages={105--113},
  year={2002},
  organization={Springer}
}

@inproceedings{o2005every,
  title={Every decision tree has an influential variable},
  author={O'Donnell, Ryan and Saks, Michael and Schramm, Oded and Servedio, Rocco A},
  booktitle={46th Annual IEEE Symposium on Foundations of Computer Science (FOCS'05)},
  pages={31--39},
  year={2005},
  organization={IEEE}
}

@article{dietmar1992randomized,
  title={On the randomized complexity of monotone graph properties},
  author={Dietmar, Gr{\"o}ger Hans},
  journal={Acta Cybernetica},
  volume={10},
  number={3},
  pages={119--127},
  year={1992}
}

@article{lovasz2002lecture,
  title={Lecture notes on evasiveness of graph properties},
  author={Lovasz, Laszlo and Young, Neal E},
  journal={arXiv preprint cs/0205031},
  year={2002}
}

@book{lovasz2009matching,
  title={Matching theory},
  author={Lov{\'a}sz, L{\'a}szl{\'o} and Plummer, Michael D},
  volume={367},
  year={2009},
  publisher={American Mathematical Soc.}
}

@article{alweiss2021subgraph,
  title={On the subgraph query problem},
  author={Alweiss, Ryan and Hamida, Chady Ben and He, Xiaoyu and Moreira, Alexander},
  journal={Combinatorics, Probability and Computing},
  volume={30},
  number={1},
  pages={1--16},
  year={2021},
  publisher={Cambridge University Press}
}

@inproceedings{matula1972employee,
  title={Employee party problem},
  author={Matula, David W},
  booktitle={Notices of the American Mathematical Society},
  volume={19},
  number={2},
  pages={A382--A382},
  year={1972},
  organization={AMER Mathematical SOC 201 Charles ST, Providence, RI 02940-2213}
}
\clearpage
\begin{appendices}


\section{Proof of Lemma \ref{max_indep_lemma}}
\label{appendix_proof_max_indep}

\begin{proof}
    Let $\mathcal{A}$ be a $(\de,l,\al)$-algorithm. We have:
    \begin{align*}
        \frac{1}{2}&\leq
        \mathbb{P}(\mathcal{A}\text{ succeeds}) \leq\sum_{p\in \mathcal{S}} \mathbb{P}(\mathcal{A}_p\text{ succeeds}) =\\
        &=\sum_{p\in \mathcal{S}_{I_{l-1}\leq \de\log n}} \mathbb{P}(\mathcal{A}_p\text{ succeeds})+\sum_{p\in \mathcal{S}_{I_{l-1}> \de\log n}} \mathbb{P}(\mathcal{A}_p\text{ succeeds})
    \end{align*}
    Assume for contradiction that $\sum_{p\in \mathcal{S}_{I_{l-1}> \de\log n}} \mathbb{P}(\mathcal{A}_p\text{ succeeds})\geq \frac{1}{4}$. Let $\mathcal{A}^{>}$ be an algorithm with the same queries as $\mathcal{A}$, that succeeds if $\mathcal{A}_p$ succeeds for some $p\in \mathcal{S}_{I_{l-1}> \de\log n}$. We have:
    \[\mathbb{P}(\mathcal{A}^>\text{ succeeds})=\sum_{p\in \mathcal{S}_{I_{l-1}> \de\log n}} \mathbb{P}(\mathcal{A}_p\text{ succeeds})\geq \frac{1}{4}\]
    Let $\mathcal{A}_1$ be the (randomized) 1-round algorithm that queries the same pairs as $\mathcal{A}^>$ queries in round $l$. $\mathcal{A}_1$ succeeds if it finds a clique of size $I_{l-1}$ for some $I_{l-1}>\de\log n$. Since $\mathcal{A}^{>}$ succeeds only if an independent set of size $I_{l-1}$ in $K_{l-1}$ is a clique in $K_l$, we have that $\mathcal{A}^{>}$ succeeds only if $\mathcal{A}_1$ succeeds, so:
    \[\mathbb{P}(\mathcal{A}_1\text{ succeeds})\geq\mathbb{P}(\mathcal{A}^>\text{ succeeds})\geq \frac{1}{4}\]
    Denote $\al_1=\frac{I_{l-1}}{\log n}$, and assume $\al_1>\de$. In the proof of Lemma \ref{cons_lemma} we get a contradiction after assuming there is a $(\de,l,\al)$-algorithm when $\al>\max_{p\in \mathcal{S}}\min_{j}\frac{2f_1^j(p)}{f_2^j(p)}$, which in our case means $\al_1>\de$. $\mathcal{A}_1$ is not a $(\de,1,\al_1)$-algorithm, because it finds a clique of size $\al_1 \log n$ with probability $\frac{1}{4}$ and not $\frac{1}{2}$. However, replacing the constant $\frac{1}{2}$ with $\frac{1}{4}$ in the proof of Lemma \ref{cons_lemma} leads to the same contradiction. 
    
    Therefore, we have $\sum_{p\in \mathcal{S}_{I_{l-1}> \de\log n}} \mathbb{P}(\mathcal{A}_p\text{ succeeds})< \frac{1}{4}$, and we get:
    \[\frac{1}{4}\leq \sum_{p\in \mathcal{S}_{I_{l-1}\leq \de\log n}} \mathbb{P}(\mathcal{A}_p\text{ succeeds})\]
    Applying the proof of Lemma $\ref{cons_lemma}$, again with $\frac{1}{4}$ instead of $\frac{1}{2}$, we get:
    \[\al_*(\de,l)\leq \max_{p\in \mathcal{S}_{I_{l-1}\leq \de\log n}}\min_{j}\frac{2f_1^j(p)}{f_2^j(p)}\]
    
\end{proof}

\section{Theorem \ref{two_rounds_big_delta_restricted_thm} calculations}
\label{appendix_two_rounds_calcs}
\begin{lemma}
Let $\frac{6}{5}<\de<2$, and denote $m^*\coloneqq\frac{1-\frac{\de}{2}}{1+\frac{\de}{2}}$, and $\al_2^{\de}(m)\coloneqq\frac{\de-(\frac{5}{2}\de-3)m}{4m^2-2m+1}$. The maximum of $\al_2^{\de}(m)$ over the domain $0\leq m\leq m^*$ is obtained when $m= m^*$, and in that case $\al_2^{\de}(m^*)=1+\frac{\de}{2}$. 

\end{lemma}
\begin{proof}
We start with showing that $\al_2^{\de}(m^*)=1+\frac{\de}{2}$:
\begin{align*}
\al_2^{\de}(m^*)&=
\frac{\de-(\frac{5}{2}\de-3)\frac{1-\frac{\de}{2}}{1+\frac{\de}{2}}}{4\Big(\frac{1-\frac{\de}{2}}{1+\frac{\de}{2}}\Big)^2-2\Big(\frac{1-\frac{\de}{2}}{1+\frac{\de}{2}}\Big)+1}=
\frac{\de(1+\frac{\de}{2})-(\frac{5}{2}\de-3)(1-\frac{\de}{2})}{4(1-\frac{\de}{2})^2-2((1-\frac{\de}{2})(1+\frac{\de}{2}))+(1+\frac{\de}{2})^2}(1+\frac{\de}{2})=\\
&=\frac{\de+\frac{1}{2}\de^2+\frac{5}{4}\de^2-4\de+3}
{\de^2-4\de+4+\frac{1}{2}\de^2-2+\frac{1}{4}\de^2+\de+1}
(1+\frac{\de}{2})=
\frac{\frac{7}{4}\de^2-3\de+3}{\frac{7}{4}\de^2-3\de+3}(1+\frac{\de}{2})=1+\frac{\de}{2}
\end{align*}
Next, we want to show this value is maximal. Since $4m^2-2m+1>0$ for any $m$, $\al_2^{\de}(m)$ is continuous and therefore the maximum over the interval $[0,m^*]$ is obtained either on the boundaries or on a local maximum. For $m=0$ we have $\al_2^{\de}(0)=\de<1+\frac{\de}{2}$, so proving there are no roots of the derivative in the interval $(0,m^*)$ will prove the lemma. 

We calculate the derivative of $\al_2^{\de}(m)$ by $m$:
\[\frac{\partial \al_2^{\de}(m)}{\partial m}=\frac{-(\frac{5}{2}\de-3)(4m^2-2m+1)-(8m-2)(\de-(\frac{5}{2}\de-3)m)}{(4m^2-2m+1)^2}=\frac{4(\frac{5}{2}\de-3)m^2-8\de m+3-\frac{1}{2}\de}{(4m^2-2m+1)^2}\]
$\frac{\partial \al_2^{\de}(m)}{\partial m}=0$ is obtained for:
\[m=\frac{8\de\pm\sqrt{(8\de)^2-16(\frac{5}{2}\de-3)(3-\frac{1}{2}\de)}}{8(\frac{5}{2}\de-3)}=\frac{4\de\pm\sqrt{21\de^2-36\de+36}}{2(5\de-6)}\]
Since $5\de-6>0$, the smaller root is the one with minus sign, so it is enough to prove $\frac{4\de-\sqrt{21\de^2-36\de+36}}{2(5\de-6)}\geq m^*$. Assuming for contradiction the opposite we get:
\begin{align*}
    \frac{4\de-\sqrt{21\de^2-36\de+36}}{2(5\de-6)}<& \frac{1-\frac{\de}{2}}{1+\frac{\de}{2}}\\
    (2+\de)(4\de-\sqrt{21\de^2-36\de+36})<&(2-\de)(10\de-12)\\
    14\de^2-24\de+24<&(2+\de)\sqrt{21\de^2-36\de+36}\\
    \sqrt{21\de^2-36\de+36}<&\frac{3}{2}(2+\de)\\
    21\de^2-36\de+36<&\frac{9}{4}(2+\de)^2\\
    75\de^2-180\de+108<&0\\
    3(5\de-6)^2<&0
\end{align*}
Which is a contradiction.

\end{proof}


\section{Optimization method}
\label{appendix_optimization}
In this section we discuss the optimization problems described in Section \ref{sec_three_rounds}. We start with the problem in the proof of Theorem \ref{thm_three_rounds} (for the problem regarding the case $q_1=q_2=q_3$ see Appendix \ref{appendix_optimization_non_restriceted}). For every $1\leq \de<2$ we have: \[\al_*((1-\frac{\de}{2},1,\de),3)\leq \max_{s_1,\st,\dt}\min\{\al_2,\al_3\}\]
where the bounds $\al_2,\al_3$ are as follows:
\begin{gather*}
    \al_3 = \frac{\de-(3\de-2)s_1-(2\de-2)(\st+\dt)}{1-2[2(\frac{1}{2}-(s_1+\st+\dt))(s_1+\st)+2s_1(\st+\dt)+\big[\frac{1-\frac{\de}{2}}{2+\de}-s_1\big](\st+2\dt)]}\\
    \al_2 = \frac{2-(2+\de)s_1-2(\st+\dt)}{1-2[2s_1(\st+\dt)+\big[\frac{1-\frac{\de}{2}}{2+\de}-s_1\big](\st+2\dt)]} 
\end{gather*}
and we have the following constraints:
\begin{gather*}
    s_1,\st,\dt\geq 0 \\
    s_1+\st+\dt\leq \frac{1}{2}\\
    2s_1+\st+2\dt\geq \stg
\end{gather*}
First, we note that for every $1\leq\de<2$, the signature $(s_1,\st,\dt)=(0,\stg,0)$ (which is also the signature of Algorithm 4) satisfies the constraints and the corresponding $\al_2,\al_3$ both equal $1+\frac{\de}{2}$. We will show that this signature maximizes $\min\{\al_2,\al_3\}$. We look at $\al_2,\al_3$ as functions of the variables $s_1,\st,\dt,\de$ over the subset of $\mathbb{R}^4$ satisfying the constraints. We prove the following lemma:
\begin{lemma}
\label{optimization lemma}
The only points in the domain of $\al_2,\al_3$ for which \[\min\{\al_2(s_1,\st,\dt,\de), \al_3(s_1,\st,\dt,\de)\}\geq 1+\frac{\de}{2}\]
are the points along the path $\gamma:[1,2]\rightarrow \mathbb{R}^4$, $\gamma(\de)=(0,\stg,0,\de)$, and the point $(s_1,\st,\dt,\de)=(0,0,\frac{1}{2},1)\eqqcolon v_{\dt=\frac{1}{2}}$ (for which $\al_2=\al_3=\frac{3}{2}$ as well).
\end{lemma}
{
\begin{remark}
Note that since we prove Lemma \ref{optimization lemma} using a computer program and not analytically, we only prove it up to a certain precision (of roughly $10^{-6}$, which is the maximal error of the optimization algorithm we use).
\end{remark}
}

The proof is divided into two phases. We sketch the proof process, and elaborate in the following subsections. First, we show that for any point which is not close to either the path $\gamma$ or the point $v_{\dt=\frac{1}{2}}$, at least one of $\{\al_2,\al_3\}$ satisfies $\al_i<1+\frac{\de}{2}$. We prove this by finding Lipschitz bounds on the functions $\al_2,\al_3$ and then evaluating the functions on a dense enough net of points over the whole domain (the density of the net depends on the Lipschitz bounds, and it is higher in areas where $1+\frac{\de}{2}-\min\{\al_2,\al_3\}$ is small). In the second phase, we look at neighborhoods of $\gamma$ and $v_{\dt=\frac{1}{2}}$. We go over all points in a dense net over these neighborhoods, and for each point use a sequential least squares programming algorithm to find a local maximum of $\min\{\al_2,\al_3\}$ around the point. All local maximum points are either the point $v_{\dt=\frac{1}{2}}$ or on the path $\gamma$, hence all points in the domain satisfy $\min\{\al_2,\al_3\}\leq 1+\frac{\de}{2}$. 

\subsection{Phase 1 - Lipschitz bounds}
We start by proving Lipschitz bounds on the functions $\al_2,\al_3$. Given a function $f:\mathbb{R}^d\rightarrow\mathbb{R}$ for which  $\forall x\in \mathbb{R}^d:\, \norm{\nabla{f}(x)}\leq L$, the mean value theorem implies $\abs{f(x)-f(y)}\leq L\norm{x-y}$. Therefore we prove upper bounds on $\norm{\nabla{\al_2}(s_1,\st,\dt,\de)}$ and $\norm{\nabla{\al_3}(s_1,\st,\dt,\de)}$, starting with the former. We have:
\begin{align*}
    \al_2(s_1,\st,\dt,\de)&= \frac{2-(2+\de)s_1-2(\st+\dt)}{1-2[2s_1(\st+\dt)+\big[\frac{1-\frac{\de}{2}}{2+\de}-s_1\big](\st+2\dt)]} \\
    &= \frac{2-(2+\de)s_1-2(\st+\dt)}{1-2[2s_1\st+2s_1\dt+\frac{1}{2}\cdot\stg\st+\stg \dt-s_1\st-2s_1\dt]} \\
    &= \frac{2-(2+\de)s_1-2(\st+\dt)}{1-2s_1\st-\stg\st-2\cdot\stg\dt}\eqqcolon\frac{f}{g}
\end{align*}
The function $f$ is bounded since $f\in[0,2]$, and the function $g$ equals $1-e_{free}^2$ (assuming $m_1'=\frac{1-\frac{\de}{2}}{2+\de}-s_1$) so by Lemma \ref{free_edges_at_most_half} we have $g\in[\frac{1}{2},1]$. We compute the partial derivatives of $f,g$ and bound them according to the constraints on the variables:
\begin{itemize}
\begin{minipage}{0.45\linewidth}
    \item $\frac{\partial f}{\partial s_1} = -(2+\de)\in[-4,-3]$
    \item $\frac{\partial f}{\partial \st} = -2$
    \item $\frac{\partial f}{\partial \dt} = -2$
    \item $\frac{\partial f}{\partial \de} = -s_1\in[-\frac{1}{2},0]$
\end{minipage}
\begin{minipage}{0.45\linewidth}
    \item $\frac{\partial g}{\partial s_1} = -2\st \in[-1,0]$
    \item $\frac{\partial g}{\partial \st} = -2s_1-\stg \in[-\frac{4}{3},0]$
    \item $\frac{\partial g}{\partial \dt} = -2\cdot\stg \in[-\frac{2}{3},0]$
    \item $\frac{\partial g}{\partial \de} = 0$
\end{minipage}
\end{itemize}
The function $\al_2$ is differentiable on the whole domain (since $g\geq \frac{1}{2}$) and we have the following bounds on its partial derivatives (using interval arithmetic):
\begin{itemize}
    \item $\abs{\frac{\partial \al_2}{\partial s_1}} = \abs{\frac{\frac{\partial f}{\partial s_1}g-\frac{\partial g}{\partial s_1}f}{g^2}}\in\abs{\frac{[-4,-3]\cdot[\frac{1}{2},1]-[-1,0]\cdot[0,2]}{[\frac{1}{4},1]}}\subseteq\abs{\frac{[-4,-\frac{3}{2}]+[0,2]}{[\frac{1}{4},1]}}\subseteq[0,16]$
    \item $\abs{\frac{\partial \al_2}{\partial \st}} = \abs{\frac{\frac{\partial f}{\partial \st}g-\frac{\partial g}{\partial \st}f}{g^2}}\in\abs{\frac{-2\cdot[\frac{1}{2},1]-[-\frac{4}{3},0]\cdot[0,2]}{[\frac{1}{4},1]}}\subseteq\abs{\frac{[-2,-1]+[0,\frac{8}{3}]}{[\frac{1}{4},1]}}\subseteq[0,8]$
    \item $\abs{\frac{\partial \al_2}{\partial \dt}} = \abs{\frac{\frac{\partial f}{\partial \dt}g-\frac{\partial g}{\partial \dt}f}{g^2}}\in\abs{\frac{-2\cdot[\frac{1}{2},1]-[-\frac{2}{3},0]\cdot[0,2]}{[\frac{1}{4},1]}}\subseteq\abs{\frac{[-2,-1]+[0,\frac{4}{3}]}{[\frac{1}{4},1]}}\subseteq[0,8]$
    \item $\abs{\frac{\partial \al_2}{\partial \de}} = \abs{\frac{\frac{\partial f}{\partial \de}}{g}}\in\abs{\frac{[-\frac{1}{2},0]}{[\frac{1}{2},1]}}\subseteq[0,1]$
\end{itemize}
Therefore, we get the following bound on $\norm{\nabla{\al_2}(s_1,\st,\dt,\de)}$:
\[\norm{\nabla{\al_2}(s_1,\st,\dt,\de)}\leq \sqrt{16^2+8^2+8^2+1^2}<19.7\eqqcolon L_2\]

We repeat the same calculations for $\al_3$:

\begin{gather*}
    \al_3(s_1,\st,\dt,\de)= \frac{\de-(3\de-2)s_1-(2\de-2)(\st+\dt)}{1-2[2(\frac{1}{2}-(s_1+\st+\dt))(s_1+\st)+2s_1(\st+\dt)+\big[\frac{1-\frac{\de}{2}}{2+\de}-s_1\big](\st+2\dt)]} \\
    = \frac{\de-(3\de-2)s_1-(2\de-2)(\st+\dt)}{1-2[s_1+\st-2s_1^2-4s_1\st-2\st^2-2\dt s_1-2\dt\st+2s_1\st+2\dt s_1+\frac{1}{2}\cdot\stg\st+\stg\dt-s_1\st-2\dt s_1]} \\
    = \frac{\de-(3\de-2)s_1-(2\de-2)(\st+\dt)}{1+4s_1^2+4\st^2+6s_1\st+4\dt s_1+4\dt\st-2s_1-\frac{6+\de}{2+\de}\st-2\cdot\stg\dt}\eqqcolon\frac{f}{g}
\end{gather*}

\begin{itemize}
\begin{minipage}{0.53\linewidth}
    \item $f\in[0,2]$
    \item $\frac{\partial f}{\partial s_1} = -(3\de-2)\in[-4,-1]$
    \item $\frac{\partial f}{\partial \st} = -(2\de-2)\in[-2,0]$
    \item $\frac{\partial f}{\partial \dt} = -(2\de-2)\in[-2,0]$
    \item $\frac{\partial f}{\partial \de} = 1-3s_1-2(\st+\dt)\in[-\frac{1}{2},1]$
\end{minipage}
\begin{minipage}{0.5\linewidth}
    \item $g\in[\frac{1}{2},1]$
    \item $\frac{\partial g}{\partial s_1} = 8s_1+6\st+4\dt-2 \in[-2,2]$
    \item ${\frac{\partial g}{\partial \st} = 8\st+6s_1+4\dt-\frac{6+\de}{2+\de} \in[-\frac{7}{3},2]}$
    \item $\frac{\partial g}{\partial \dt} = 4s_1+4\st-2\cdot\stg \in[-\frac{2}{3},2]$
    \item $\frac{\partial g}{\partial \de} = 0$
\end{minipage}
\end{itemize}

\begin{itemize}
    \item $\abs{\frac{\partial \al_3}{\partial s_1}} = \abs{\frac{\frac{\partial f}{\partial s_1}g-\frac{\partial g}{\partial s_1}f}{g^2}}\in\abs{\frac{[-4,-1]\cdot[\frac{1}{2},1]-[-2,2]\cdot[0,2]}{[\frac{1}{4},1]}}\subseteq\abs{\frac{[-4,-\frac{1}{2}]+[-4,4]}{[\frac{1}{4},1]}}\subseteq[0,32]$
    \item $\abs{\frac{\partial \al_3}{\partial \st}} = \abs{\frac{\frac{\partial f}{\partial \st}g-\frac{\partial g}{\partial \st}f}{g^2}}\in\abs{\frac{[-2,0]\cdot[\frac{1}{2},1]-[-\frac{7}{3},2]\cdot[0,2]}{[\frac{1}{4},1]}}\subseteq\abs{\frac{[-2,0]+[-4,\frac{14}{3}]}{[\frac{1}{4},1]}}\subseteq[0,24]$
    \item $\abs{\frac{\partial \al_3}{\partial \dt}} = \abs{\frac{\frac{\partial f}{\partial \dt}g-\frac{\partial g}{\partial \dt}f}{g^2}}\in\abs{\frac{[-2,0]\cdot[\frac{1}{2},1]-[-\frac{2}{3},2]\cdot[0,2]}{[\frac{1}{4},1]}}\subseteq\abs{\frac{[-2,0]+[-4,\frac{4}{3}]}{[\frac{1}{4},1]}}\subseteq[0,24]$
    \item $\abs{\frac{\partial \al_3}{\partial \de}} = \abs{\frac{\frac{\partial f}{\partial \de}}{g}}\in\abs{\frac{[-\frac{1}{2},1]}{[\frac{1}{2},1]}}\subseteq[0,2]$
\end{itemize}
\[\norm{\nabla{\al_3}(s_1,\st,\dt,\de)}\leq \sqrt{32^2+24^2+24^2+2^2}<46.7\eqqcolon L_3\]

The bounds we proved on the gradients of $\al_2,\al_3$ imply that any points $x,y$ in the domain for which $\norm{x-y}\leq \epsilon$, satisfy $\abs{\al_i(x)-\al_i(y)}\leq L_i\epsilon$. Therefore, if for some point $y$ we have that for at least one $i\in\{2,3\}$, $\al_i(y) + L_i\epsilon\leq 1+\frac{\de}{2}$, we deduce that for any $x$ for which $\norm{x-y}\leq \epsilon$ we have $\min\{\al_2(x),\al_3(x)\}\leq 1+\frac{\de}{2}$. 

We partitioned the domain into $4$-dimensional hypercubes with edges of size $\epsilon$, and calculated the values of $\al_2,\al_3$ on an $\epsilon$-net of points which includes all centers of the hypercubes. Given a point $x$ in the domain, let $y$ be the center of the hypercube containing $x$. The distance between $x$ and $y$ on each coordinate is at most $\frac{\epsilon}{2}$. Hence, $\norm{x-y}\leq \sqrt{4\Big(\frac{\epsilon}{2}\Big)^2}=\epsilon$. For any center $y$ which satisfies $\al_i(y) + L_i\epsilon\leq 1+\frac{\de}{2}$ for some $i\in\{2,3\}$, we deduce that all points in the hypercube satisfy $\min\{\al_2(x),\al_3(x)\}\leq 1+\frac{\de}{2}$. For $\epsilon=0.01/L_3$ this process eliminates most area of the domain, and all remaining hypercubes are close to either $\gamma$ or $v_{\dt=\frac{1}{2}}$ (the distance is at most $0.041$). Since all centers satisfy $\al_i(y) \leq 1+\frac{\de}{2}$ for some $i\in\{2,3\}$, this phase is enough for proving the weaker claim: $\forall x:\, \min\{\al_2(x),\al_3(x)\}\leq 1+\frac{\de}{2}+0.01$. 

We chose $\epsilon=0.01/L_3$ in order to eliminate a large fraction of the domain. However, an $\epsilon$-net of the domain includes roughly $10^{13}$ points, and evaluating $\al_2,\al_3$ on all of them is computationally heavy. Instead, we used a search tree to create a net in which the density of points in some area depends on how close $\al_2,\al_3$ are to $1+\frac{\de}{2}$ in that area. We start with a sparse $\epsilon_0$-net for some large $\epsilon_0$. For any hypercube with edge size $\epsilon_i$ (starting with initial ones of size $\epsilon_0$), we evaluate $\al_2,\al_3$ on the center $y$, and if $\al_i(y) + L_i\epsilon_i\leq 1+\frac{\de}{2}$, we deduce that all points in the hypercube may be eliminated (we do not need a denser net in this area). Otherwise, we partition the hypercube into $16$ hypercubes with edge size $\epsilon_{i+1}\coloneqq \frac{\epsilon_i}{2}$ and process them recursively (near the boundary of the domain we may partition into less than $16$ hypercubes). In the search tree this means we either reach a leaf (i.e., we stop partitioning the relevant path), or we create $16$ children of the current hypercube and continue the search on their subtrees. When we reach a hypercube of edge size $\epsilon_i\leq \epsilon$ which needs to be partitioned, we do not partition it further because we reached maximum density. This process is more efficient, and it eliminates roughly the same area as a full $\epsilon$-net. 

\subsection{Phase 2 - sequential least squares programming}

In Phase 2 we prove that in any remaining hypercube all points satisfy $\min\{\al_2,\al_3\}\leq 1+\frac{\de}{2}$. For each hypercube we use a sequential least squares programming (SLSQP) algorithm which maximizes a non linear function under non linear constraints. Since SLSQP finds a local maximum, we do not use it over the whole domain, but rather on small hypercubes (small in relation to the second derivatives of the functions). The SLSQP algorithm works when the objective function and the constraints are twice continuously differentiable. Our objective function is not differentiable (minimum of two functions) so we modify the optimization problem before applying SLSQP. Any point $x\in\gamma$ satisfies $\al_2(x)=\al_3(x)=1+\frac{\de}{2}$. Therefore, in order to find $\max_{s_1,\st,\dt}\min\{\al_2,\al_3\}$ for every $\de$, it is enough to limit the search only to points in which $\al_2\geq 1+\frac{\de}{2}$ and $\al_3\geq 1+\frac{\de}{2}$ (we used the same argument on $\al_1$). We get the following constraints:
\begin{gather*}
    1+\frac{\de}{2}\leq \frac{\de-(3\de-2)s_1-(2\de-2)(\st+\dt)}{1-2[2(\frac{1}{2}-(s_1+\st+\dt))(s_1+\st)+2s_1(\st+\dt)+\big[\frac{1-\frac{\de}{2}}{2+\de}-s_1\big](\st+2\dt)]}\\
    1+\frac{\de}{2}\leq \frac{2-(2+\de)s_1-2(\st+\dt)}{1-2[2s_1(\st+\dt)+\big[\frac{1-\frac{\de}{2}}{2+\de}-s_1\big](\st+2\dt)]} 
\end{gather*}
We use SLSQP to show that the only points which satisfy all constraints are either the point $v_{\dt=\frac{1}{2}}$ or they are on the path $\gamma$. This will prove Lemma \ref{optimization lemma}.
Given a remaining hypercube close to $\gamma$, and assuming a point which satisfies all constraints exists in it, we want to guarantee it will be found by the SLSQP algorithm, when it is initialized in the center of the hypercube. We create a maximization problem with the previous and new constraints, where the objective function is the square of the distance from the path $\gamma$:
\[\max_{s_1,\st,\dt,\de} s_1^2+(\st-\frac{2-\de}{2+\de})^2+\dt^2\]
If the algorithm returns a point on $\gamma$ (i.e., with distance $0$ from the path) this means that no point in the hypercube satisfies all constraints (unless the hypercube intersects the path). Equivalently, for remaining hypercubes close to $v_{\dt=\frac{1}{2}}$, the maximized objective function will be the square of the distance from $v_{\dt=\frac{1}{2}}$:
\[\max_{s_1,\st,\dt,\de} s_1^2+\st^2+(\frac{1}{2}-\dt)^2+(1-\de)^2\]

For each hypercube remaining after Phase 1 (roughly $2\cdot 10^8$ hypercubes) we applied SLSQP initialized in its center. No run of the algorithm failed to find a maximizing point, and in all runs this point was on the path $\gamma$ (or it was $v_{\dt=\frac{1}{2}}$ for the hypercubes in its neighborhood). Therefore, there is no other point in any of the remaining hypercubes which satisfies all constraints, and Lemma \ref{optimization lemma} follows.

\subsection{Optimization method for non-restricted case}
\label{appendix_optimization_non_restriceted}
We solve the optimization problem described in Section \ref{sec_three_rounds_non_restricted} for any $\de$ of the form $\de=1+0.1z$ for the integers $0\leq z\leq 9$. Since we do not prove an upper bound {that needs to hold for all} $\de$, we skip phase 1, and instead use SLSQP in order to find all local maximum points of the objective function. For $\epsilon=0.01$, we create an $\epsilon$-net spanning the domain of the parameters $s_1,\st,\dt,m_1'$, and for each point in the net run SLSQP initialized in the point. Then we take the maximum over all returned values. Since the objective function is not twice continuously differentiable, we introduce a new variable $t$ and solve the following equivalent optimization problem:
\[\max t\]
Under the constraints:
\begin{itemize}
    \item $t\leq \frac{\de}{1-2[2(\frac{1}{2}-(s_1+\st+\dt))(s_1+\st)+2s_1(\st+\dt)+m_1'(\st+2\dt)]}$
    \item $t\leq \frac{2-(4-2\de)(s_1+\st+\dt)}{1-2[2s_1(\st+\dt)+m_1'(\st+2\dt)]}$
    \item $t\leq2-(4-2\de)(s_1+m_1') $
    \item $s_1,\st,\dt,m_1'\geq 0$
    \item $s_1+\st+\dt\leq \frac{1}{2}$
    \item $m_1'\leq \frac{\st}{2}+\dt$
    \item $2s_1+\st+2\dt\geq \stg$
\end{itemize}
\end{appendices}

\end{document}